\documentclass[11pt]{article}

\usepackage{siunitx}
\usepackage[T1]{fontenc}
\usepackage{verbatim}
\usepackage{amsmath,amsthm}
\usepackage{xspace}
\usepackage{pifont}
\usepackage{amssymb}
\usepackage{dsfont}
\usepackage[table]{xcolor}
\usepackage{arydshln}
\usepackage{makeidx}
\usepackage{mathrsfs}
\usepackage{xcolor, colortbl}
\usepackage{subcaption}
\usepackage[numbers]{natbib}
\usepackage{enumitem}
\usepackage{booktabs}
\usepackage{afterpage}
\usepackage{mathtools}
\usepackage{graphicx}
\usepackage{float}
\usepackage[export]{adjustbox}
\usepackage{tcolorbox}
\usepackage{authblk}
\usepackage[normalem]{ulem}
\usepackage[colorlinks,citecolor=blue,urlcolor=blue]{hyperref}
\mathtoolsset{showonlyrefs}

\def\eps{\varepsilon}

\def\P{\mathbb{P}}

\newcommand{\R}{{\mathbb R}}

\newcommand{\N}{{\mathbb N}}

\definecolor{ablightgreen}{RGB}{226,246,226}

\newcommand{\bd}{\begin{displaymath}}
\newcommand{\ed}{\end{displaymath}}
\newcommand{\be}{\begin{equation}}
\newcommand{\ee}{\end{equation}}
\newcommand{\bq}{\begin{eqnarray}}
\newcommand{\eq}{\end{eqnarray}}
\newcommand{\bn}{\begin{eqnarray*}}
\newcommand{\en}{\end{eqnarray*}}

\newtheorem{theorem}{Theorem}[section]

\newtheorem{remark}[theorem]{Remark}

\numberwithin{equation}{section}

\usepackage{geometry}
 \newcommand*\samethanks[1][\value{footnote}]{\footnotemark[#1]}

\begin{document}

\title{Optimal Execution with Passive Market Impact}
\author[1]{Alexander Barzykin }
\author[2]{Robert Boyce \thanks{RB and ST are supported by the EPSRC Centre for Doctoral Training in Mathematics of Random \mbox{Systems}: Analysis, Modelling and Simulation (EP/S023925/1).}}
\author[2]{Eyal Neuman }
\author[2]{Sturmius Tuschmann \samethanks } 
\affil[1]{HSBC}
\affil[2]{Department of Mathematics, Imperial College London}
\date{\today}
\maketitle

\begin{abstract}
We derive a mesoscopic model for optimal execution with limit orders that incorporates microstructural features of passive price impact. Our framework is based on two empirical observables: the approximately exponential decay of limit-order fill probabilities with distance from the midprice, and the short-term linear response of price changes to order flow imbalance. Combining these ingredients, we obtain a reduced-form passive impact rate that decays exponentially with quote distance. The model describes passive execution at a tactical level, where fills arise from a sequence of quote adjustments that balance execution probability, adverse selection, and opportunity cost. We formulate and solve an optimal liquidation problem in which the trader controls the aggressiveness of passive sell quotes. This generates a trade-off between higher fill intensity and larger accumulated impact on the one hand, and lower impact but greater non-execution risk on the other. Empirical calibration using NASDAQ equities and public FX supports the empirical foundations of the model. We also analyse extensions with heterogeneous decay rates, transient impact, and target execution schedules.
\end{abstract}

\begin{description}
\item[Mathematics Subject Classification (2020):]  93E20, 60H30, 91G80
\item[JEL Classification:] C02, C61, G11

\item[Keywords:]  optimal trading, market impact, passive impact, limit orders 
\end{description}


\section{Introduction}\label{sec:introduction}

Price impact refers to the empirical observation that executing a large order adversely affects the price of a risky asset in a persistent manner, resulting in less favourable execution prices. Consequently, an agent seeking to liquidate a large position, known as a metaorder, typically splits it into smaller components, referred to as child orders, which are executed over a period of hours or days. Price impact is therefore a central ingredient in determining the optimal execution schedule of such an order \cite{bouchaud_bonart_donier_gould_2018,Cartea2015Algorithmic,Gueant2016Financial,Webster2023Handbook}.
Classical execution theories, such as those of \citet{Bertsimas1998Optimal} and \citet{Almgren2001Optimal}, model trading costs as the price response generated by an investor's own market orders. In this framework, the trader controls an execution rate and faces a trade-off between expected impact costs and inventory risk. Subsequent work introduced transient liquidity effects, market resilience, no-dynamic-arbitrage constraints and utilisation of short-term alphas (see, e.g., \citep{abijaber2026optimal,Gatheral2010Dynamic,lehalle2019signals,neuman2022optimal,Obizhaeva2013Optimal}).

In parallel, a large empirical literature has been devoted to the estimation of price impact at the metaorder level. Across asset classes and datasets, empirical studies have found that price impact is typically concave in executed volume and depends on participation rate, execution duration, and market liquidity in a manner often summarised by square-root or related power-law specifications \citep{Bacry2015Market,Bershova2013Nonlinear,Donier_15,hey2023concave,Lillo2003Master,Sato2025Strict,toth2016squarerootimpactlawholds,Toth_17,Zarinelli2015Beyond}. More recent work has focused on the calibration of cross-impact and on further developments in the nonparametric estimation of self- and cross-impact \citep{benzaquen2017dissecting,hey2024cross,Hey2025Nonparametric,coz2023cross,Mastromatteo2017,neuman2023statistical,neuman2023offline,schneider2019cross,tomas2022build}.
A notable feature of these empirical laws and their microstructural explanations is that they are formulated for metaorders rather than individual child orders. This observation suggests that the natural object of study for execution theory is not an isolated transaction, but rather the aggregate trading policy through which inventory is liquidated.

Most optimal execution models nevertheless associate price impact most directly with liquidity-taking trades. Passive execution is often introduced through a fill intensity term: a limit order posted farther from the midprice earns a better price conditional on execution but is less likely to be filled. This is the mechanism underlying the market-making model of \citet{Avellaneda2008High} and the high-frequency trading framework of \citet{Cartea2014Buy} (see also \citep{Gueant2012Optimal,Guilbaud2013Optimal,Huang2015Simulating}). In these models, the passive order affects the trader’s cash and inventory, but the fill itself is usually not endowed with an explicit permanent price effect.

\citet{HautschHuang2012} provide an econometric model showing that limit orders in an order book generate both short-run and long-run price impacts, estimated using Euronext Amsterdam stock data. The magnitude and direction of the impact depend on order aggressiveness, size, and the state of the order book, with cross-sectional variation largely explained by trading frequency and tick size. A microscopic model for such effects was provided by \citet{Cont2014Price}, who show empirically that, over short horizons, midprice changes are well explained by order flow imbalance at the best quotes, where the imbalance aggregates market orders, limit-order placements, and cancellations on both sides of the book. \citet{ouazzani2024theory} develop a model in which the insertion and cancellation of limit orders affect the intensity of market order arrivals in the best bid and ask queues in the limit order book. This, in turn, affects the order-book imbalance, which in turn moves the midprice. This is the closest theoretical antecedent to the present work and provides theoretical support for the idea that passive execution can have a systematic impact.

Nevertheless, the approach in \cite{ouazzani2024theory} is fundamentally an order-book model. This may not always be the natural state representation for quote-driven markets such as foreign exchange. Spot FX liquidity is fragmented across bilateral and multi-dealer venues, dealer streams, internalisation pools, last-look protocols, and credit-specific relationships. There is generally no single consolidated central order book whose queue lengths, cancellations, and priority rules summarise the market state. In this environment, a trader’s passive execution process is better described operationally as repeated interaction with dealer or venue quotes, rather than as the evolution of a transparent exchange queue. Recent work on FX internal liquidity by \citet{Barzykin2026FX} illustrates a closely related over-the-counter (OTC) origin of passive impact: when a passive flow is submitted by a client to an internal exchange, the market maker adjusts pricing in their external OTC franchise to facilitate matching flow. The dealer’s pricing and internalisation decisions therefore transmit passive liquidity demand into external quote dynamics, even before one has specified a full central order-book mechanism. A model calibrated to full placement and cancellation intensities in an equity-style order book may therefore be too detailed in the wrong state variables for FX, even though the underlying economic insight that order flow and quote-driven informational feedback move prices remains highly relevant.

The objective of this work is to retain the microstructural justification for passive impact while moving to a mesoscopic continuous formulation that is independent of the order book and also appropriate for optimal execution. The cornerstones of our model are two empirical ingredients. First, the fill intensity of limit orders decays exponentially with distance from the midprice, as in the Avellaneda–Stoikov model~\cite{Avellaneda2008High}. Second, short-horizon price changes respond linearly to order flow imbalance, as in~\citet{Cont2014Price}. Combining these ingredients in Section \ref{sec:derivation}, we derive a reduced-form passive impact rate for limit orders that decays exponentially with the quote distance and is proportional to the fill intensity. 

Our framework deliberately does not model every child order placement, cancellation, and reposting decision. Instead, it interprets passive execution at the tactical level. In practice, a single passive fill in an execution algorithm is rarely the result of a single untouched child order passively waiting in the book. It is typically achieved through a chain of quote submissions, cancellations, reposts, venue choices, and quote-distance adjustments designed to obtain a fill while controlling adverse selection and opportunity cost.
The economically relevant impact for optimal execution is therefore not the impact of an isolated passive child limit order, but rather the impact generated by the tactical strategy used to achieve a passive fill. The mesoscopic price impact term in our model ($\eta e^{-m\delta}$, where $\delta$ is the distance from the midprice and $m$ is an empirically estimated decay constant) summarises this chain by linking the trader’s chosen quote distance to both the probability of execution and the order flow imbalance conditional on execution. The special case in which the fill intensity and the passive impact term have the same exponential decay rate is of particular importance, since our NASDAQ stock estimates suggest that this is a good approximation for the equity sample. The case of different decay rates is then motivated by our FX estimates, where the gap between the two estimated decays is more visible.

The optimal execution problem for the empirically relevant equal-decay case, defined and solved explicitly in Section \ref{sec:model}, is a passive-impact analogue of the classical impact models described above. The agent liquidates their inventory by choosing the quote distance of passive sell orders. In this setup, a more aggressive quote increases the fill intensity but also increases the rate at which passive impact is accumulated, whereas a less aggressive quote reduces immediate impact but increases non-execution risk and inventory costs. The solution is further analysed in Section \ref{sec:illustrations} using empirically realistic fill intensity and price impact parameters, which are estimated in Section \ref{sec:empirics} using NASDAQ stock data and LSEG FX market data. 

We observe that the agent’s optimal quote is closer to the midprice when the remaining inventory is large, reflecting increased urgency due to risk aversion and the terminal penalty. This effect reverses toward the end of the trading period because inventory is marked to the impacted midprice, creating a situation in which the trader would prefer to incur the inventory penalty associated with an additional unit rather than bear the impact on the entire position (see Figure \ref{fig:quotes-baseline}). We also observe in Figure \ref{fig:inventory_quantilemap} that the median inventory trajectory for passive order execution is remarkably close to that of the Almgren–Chriss model, which considers only market orders. A key result of \citet{Almgren2001Optimal} is that permanent price impact does not affect the trading rate but only the optimal strategy’s P\&L. In contrast, in the passive-order case, permanent impact drives the convexity of the liquidation curve (see the detailed explanation below Figure \ref{fig:inventory_quantilemap}). 

Extensions of the model, including cases in which the fill intensity and the passive impact term have different decay rates, as well as models with transient passive impact and target execution schedules, are studied in Section~\ref{sec:extensions}. In particular, the fill intensity and passive impact decay rates are found to differ in the FX data reported in Table~\ref{tab:fx-parameter-estimates}, which makes the derivation of the optimal quote considerably more involved. One possible explanation is that this difference arises because our estimates are based only on public LSEG quote and trade data, whereas most FX trading occurs OTC, allowing traders to draw on additional sources of information. Moreover, available L2 information in FX is less detailed than L3 information in equities. Finally, Section~\ref{sec-rem} contains concluding remarks, while Sections~\ref{sec:proof-main} and \ref{sec:proofs-extensions} are devoted to the proofs of our results.

\section{Passive impact derivation}\label{sec:derivation}

Consider a stylised model for a limit order book over a short time interval \([t,t+h]\), where \(\smash{L^{b/a}_{t,t+h}}\), \(\smash{C^{b/a}_{t,t+h}}\), and \(\smash{M^{b/a}_{t,t+h}}\) denote the total size of limit order placements, limit order cancellations, and market orders in \([t,t+h]\) on the best bid and best ask levels respectively. Following \citet{Cont2014Price}, define the corresponding order flow imbalance (OFI) over the interval \([t,t+h]\) by
\begin{equation}\label{eq:OFI}
    \operatorname{OFI}_{t,t+h}
    =
    L^{b}_{t,t+h}
    -
    C^{b}_{t,t+h}
    -
    M^{a}_{t,t+h}
    -
    L^{a}_{t,t+h}
    +
    C^{a}_{t,t+h}
    +
    M^{b}_{t,t+h}.
\end{equation}
A key finding of \cite{Cont2014Price} is that changes in the midprice \(S\) over the same interval follow
\begin{equation}\label{eq:OFI-relation}
    \Delta S_{t,t+h}
    =
    \beta\,\operatorname{OFI}_{t,t+h}
    +
    \eps_{t,t+h},
\end{equation}
where
\begin{equation}\label{eq:DeltaS}
    \Delta S_{t,t+h}:=S_{t+h}-S_t,
\end{equation}
$\eps_{t,t+h}$ is a noise term capturing influences of other factors, such as deeper levels of the book, and \(\beta\) is the price impact coefficient. This linear relation captures the empirical observation of \cite{Cont2014Price} that short-term price changes are primarily driven by order flow imbalance at the best bid and ask quotes.

In the appendix of their original paper \cite{Cont2014Price}, the authors briefly extended the OFI framework to also include deeper levels of the order book. The corresponding framework of multi-level order flow imbalance (MLOFI) was made precise and analysed in detail in the empirical study by \citet{xu2018multi}. For $d=1,\ldots,\bar d$, let \(\smash{L^{d,b/a}_{t,t+h}}\), \(\smash{C^{d,b/a}_{t,t+h}}\), and \(\smash{M^{d,b/a}_{t,t+h}}\) denote the total size of limit order placements, limit order cancellations, and market orders in \([t,t+h]\) on the $d$-th bid level and $d$-th ask level respectively. Analogously to \eqref{eq:OFI}, define the corresponding MLOFI over the interval \([t,t+h]\) by
\begin{equation}\label{eq:MLOFI}
    \operatorname{MLOFI}^d_{t,t+h}
    =
    L^{d,b}_{t,t+h}
    -
    C^{d,b}_{t,t+h}
    -
    M^{d,a}_{t,t+h}
    -
    L^{d,a}_{t,t+h}
    +
    C^{d,a}_{t,t+h}
    +
    M^{d,b}_{t,t+h},
    \qquad
    d=1,\ldots,\bar d.
\end{equation}
For $\bar d=1$, MLOFI and OFI are identical. Motivated by the linear relation \eqref{eq:OFI-relation} between OFI and price changes, the authors of \cite{xu2018multi} perform a multiple ridge regression and similarly show a linear relation between MLOFI and price changes. In particular, they find that the corresponding out-of-sample goodness-of-fit improves with each additional included price level, that is, it increases in $\bar d$. Here the regression coefficients $\beta^1,\ldots,\beta^{\bar d}$ in $\mathrm{ticks}\cdot\mathrm{shares}^{-1}$ can be interpreted as price impact coefficients, where $\beta^d$ captures the partial contribution of MLOFI at the $d$-th order-book level to the price change (see \cite{xu2018multi}, Tables 8 and 9).

Now, suppose a trader places a limit order at distance \(\delta\) ticks from the midprice. \citet{Avellaneda2008High} model the execution of such an order by a Poisson process whose intensity $\Lambda(\delta)$ decreases with its distance $\delta$ from the midprice. Their specification is motivated by two empirical observations. First, market order sizes follow a power-law distribution \cite{gabaix2003theory,Gopikrishnan2000Statistical,Maslov2001Price}. Second, price changes following market orders are approximately proportional to the logarithm of market order size \cite{Potters2003More}. Combining these two observations, \citet{Avellaneda2008High} obtain a fill intensity in $\mathrm{shares}\cdot \mathrm{s}^{-1}$ of the form
\be\label{eq:execution-intensity}
    \Lambda(\delta)=\lambda e^{-k\delta},
\ee
where \(\lambda,k\) are positive constants. Thus, the expected time taken until the limit order is filled in $\mathrm{s}\cdot\mathrm{shares}^{-1}$ is equal to $\lambda^{-1} e^{k\delta}$. Therefore, returning to the MLOFI model, and working directly in terms of the posted distance from the midprice $\delta$ instead of order book level $d$, the expected price impact rate caused by posting at distance $\delta$ in $\mathrm{ticks}\cdot\mathrm{s}^{-1}$ is given by
\begin{equation} \label{eq:expected_impact_beta}
    \beta(\delta)\big/{\lambda^{-1} e^{k\delta}} = \beta(\delta)\lambda e^{-k \delta}.
\end{equation}
Motivated by the empirical results of \cite{xu2018multi}, we propose as a reduced-form specification that $\beta(\delta)$ decays approximately exponentially in $\delta$, which will also be examined empirically in Section \ref{sec:empirics}. In particular, for
\begin{equation}\label{eq:mlofi-decay}
    \beta(\delta) =\xi e^{-\ell \delta},
\end{equation}
where $\xi,\ell$ are constants, we can write the passive impact rate \eqref{eq:expected_impact_beta} as
\begin{equation} \label{eq:passive-impact}
    \xi\lambda e^{-(k+\ell)\delta} =: \eta e^{-m\delta},
\end{equation}
where the units of $\eta:=\xi\lambda$ and $m:=k+\ell$ are $\mathrm{ticks}\cdot\mathrm{s}^{-1}$ and $\mathrm{ticks}^{-1}$ respectively. In the execution model presented in Section \ref{sec:model}, the term $\eta e^{-m\delta}$ is interpreted as a mesoscopic average passive impact rate, rather than as the literal pathwise impact of one isolated child order. The remaining idiosyncratic price response is absorbed into the noise component of the price dynamics. In Section \ref{sec:empirics}, we find that for equities, $\ell$ is very small relative to $k$, so that $m\approx k$, whereas for FX, the gap is more visible, motivating the extension to distinct decay rates. The case $m=k$ corresponds to passive impact that decays in the distance to the midprice at the same rate of fill intensity.  

\begin{remark}
The above execution model can be interpreted as the continuous-time limit of a discrete sequence of quote submissions. Suppose that, in order to maintain a quote at distance \(\delta\) from the moving midprice, the trader repeatedly cancels and reposts the limit order. Using the execution intensity \eqref{eq:execution-intensity}, if each posted order remains active for a short time \(\tau>0\) then its probability of being filled during one attempt is
\[
    p_\delta
    =
    1-\exp(-\lambda e^{-k\delta}\tau)
    =
    \lambda\tau e^{-k\delta}+o(\tau).
\]
Thus the number \(G_\delta\) of submissions required until execution is geometrically distributed with parameter \(p_\delta\), and hence, for small $\tau$,
\be\label{eq:Gdelta}
    \mathbb E[G_\delta]
    =
    \frac{1}{p_\delta}
    =
    \frac{1}{\lambda\tau e^{-k\delta}}+o\bigg(\frac{1}{\tau}\bigg).
\ee
In particular, deeper quotes require exponentially more submissions on average before execution. Multiplying \eqref{eq:Gdelta} by the attempt length $\tau$ gives
\[
    \mathbb E[\tau G_\delta]
    =
    \frac{\tau}{p_\delta}
    \longrightarrow
    \lambda^{-1} e^{k\delta},
    \qquad\text{as } \tau\downarrow0,
\]
which is precisely the mean waiting time of an exponential random variable with rate \(\lambda e^{-k\delta}\). Hence the Poisson dynamics used above can be viewed as the continuous-time limit of a microscopic trial-and-error mechanism based on repeated limit-order submissions. In our model, this microscopic mechanism is summarised by the execution intensity \eqref{eq:execution-intensity}.
\end{remark}

\section{Model setup and preliminary results}\label{sec:model}

In this section, we formulate the optimal execution problem based on the passive impact rate derived in Section~\ref{sec:derivation}, and solve the case $m=k$ explicitly. 

\subsection{Model setup}

We fix a finite deterministic time horizon $T>0$ and a filtered probability space $(\Omega,\mathcal{F},(\mathcal{F}_t)_{t\in [0,T]},\P)$ satisfying the usual conditions of right-continuity and completeness. We consider a trader who seeks to liquidate $q_0\in\N$ inventory units over the finite time interval $[0,T]$ using unit sell limit orders. The distance $\delta$ to the midprice of the posted limit orders is chosen by the trader from the class of admissible strategies
\begin{equation}\label{eq:A}
    \mathcal{A}:=\bigg\{\delta:\Omega\times [0,T]\to\R\,\bigg|\, \delta\textrm{ is }(\mathcal{F}_t)_{t\in [0,T]}\textrm{-predictable, }\exists C<\infty:\ \delta\geq C\ \P\otimes dt-\textrm{a.e.}\bigg\}.
\end{equation}
In particular, negative quote distances $\delta<0$ are allowed and may be interpreted as aggressive or marketable limit orders. Let $(N_t)_{t\in [0,T]}$ be a counting process with stochastic intensity 
\begin{equation}\label{eq:execution-intensity-model}
\lambda e^{-k \delta_t}\mathbf{1}_{\{Q_{t-}>0\}},
\qquad
t\in [0,T],
\end{equation}
as in \eqref{eq:execution-intensity}, representing the cumulative fills of the trader's limit orders by time $t$, where $\lambda$ has units of inverse time for a unit order, $k$ is a positive constant, and $(Q_t)_{t\in [0,T]}$ denotes the trader's inventory given by
\be\label{eq:inventory}
    Q_t
    =
    q_{0}
    -
    N_t,
    \qquad
    t\in [0,T].
\ee
In particular, the trader stops trading after all $q_0$ shares are fully liquidated. We moreover assume that the trader's sell limit orders cause permanent passive impact until full liquidation at rate 
\be\label{eq:passive-impact-model}
\eta e^{-m \delta_t}\mathbf{1}_{\{Q_{t-}>0\}},\qquad t\in [0,T],
\ee
as in \eqref{eq:passive-impact}, where $\eta,m$ are nonnegative constants. Let $S_0\in\R$ denote the initial midprice, let $\sigma$ be a nonnegative volatility constant, and let $(W_t)_{t\in [0,T]}$ be a standard Brownian motion on $(\Omega,\mathcal{F},(\mathcal{F}_t)_{t\in [0,T]},\P)$. The impacted midprice process $(S_t)_{t\in [0,T]}$ is then given by
\begin{equation}\label{eq:S}
    S_t
    = S_{0} +
    \sigma W_t
    -
    \eta  \int_{0}^{t} e^{-m \delta_u}\mathbf 1_{\{Q_{u-}>0\}}\,du,
\end{equation}
where the drift term in \eqref{eq:S} captures the permanent passive impact specified in \eqref{eq:passive-impact-model}. The trader's cash process $(X_t)_{t\in [0,T]}$ evolves as
\be
    X_t
    =
    \int_{0}^{t}
    (S_u + \delta_u) dN_u.
\ee
Let $\mathbb{E}_{t, x, q, s}[\cdot]$ denote the expectation conditional on $X_{t}=x$, $Q_{t}=q$, $S_{t}=s$. The trader's goal is to maximise expected cash, net of running and terminal inventory penalties, over the horizon $[t,T]$:
\begin{equation}\label{eq:optimisation-problem}
     u(t,x,q,s)
    :=
    \sup_{\delta \in \mathcal A}
    \mathbb E_{t,x,q,s}
    \bigg[
        X_T + Q_T S_T 
        - \phi \int_t^T Q_u^2\,du
        - \alpha Q_T^{2}
    \bigg],
\end{equation}
where $\mathcal{A}$ is the set of admissible controls from \eqref{eq:A}, and $\phi,\alpha$ are nonnegative constants. The first term on the right-hand side of \eqref{eq:optimisation-problem} represents the trader's final cash position, the second term captures the final book value of the remaining shares valued at the impacted midprice, and the third and fourth terms represent the running and terminal inventory penalties respectively.

\subsection{HJB equation and optimal strategy}

Recall that $m=k+\ell$. As will be seen in Section \ref{sec:empirics}, the case $m=k$ is a very close approximation to reality for equity markets (where $\ell\ll k$). In this case, the optimisation problem \eqref{eq:optimisation-problem} has a particularly elegant solution. Namely, the associated Hamilton--Jacobi--Bellman (HJB) equation is
\begin{equation}
\label{eq:hjb}
\begin{aligned}
    0
    =
    &\frac{\partial u}{\partial t}(t,x,q,s)
    + \frac{1}{2}\sigma^2 \frac{\partial^{2} u}{\partial s^{2}}(t,x,q,s)
    - \phi q^2 \\
    &+ \mathbf 1_{\{q>0\}}
    \sup_{\delta}
    \bigg\{
        \lambda e^{-k\delta}
        \bigg(
             u(t,x+s+\delta,q-1,s) - u(t,x,q,s)
            - \frac{\eta}{\lambda} \frac{\partial u}{\partial s}(t,x,q,s)
        \bigg)
    \bigg\},
\end{aligned}
\end{equation}
with boundary and terminal conditions
\be
     u(t,x,0,s) = x, 
     \qquad 
     u(T,x,q,s) = x + qs - \alpha q^{2}.
\ee

\begin{theorem}\label{thm:main}
Let $m=k$. The optimal quote is given by
\begin{equation}
\label{eq:optimal-quote-omega}
    \delta^\star(t,q)
    =
    \frac{1}{k}
    +
    \frac{1}{k}
    \log\bigg(
        \frac{\omega(t,q)}{\omega(t,q-1)}
    \bigg)
    +
    \frac{\eta}{\lambda} q,
    \qquad 
    t\in [0,T],\ q=1,\dots,q_0,
\end{equation}
where
\begin{equation}
\label{eq:omega-matrix-representation}
    \begin{pmatrix}
    \omega(t,q_0)\\
    \omega(t,q_0-1)\\
    \vdots\\
    \omega(t,1)\\
    \omega(t,0)
\end{pmatrix}
    =
    \exp\big(-(T-t)A\big)
    \begin{pmatrix}
    e^{-k\alpha q_0^2}\\
    e^{-k\alpha (q_0-1)^2}\\
    \vdots\\
    e^{-k\alpha}\\
    1
\end{pmatrix},
\end{equation}
and \(A\in\mathbb R^{(q_0+1)\times(q_0+1)}\) is the upper bidiagonal matrix
\begin{equation}
\label{eq:omega-matrix-A}
A
=
\begin{pmatrix}
    q_0^2k\phi
    &
    -e^{-1}\lambda e^{-k\frac{\eta}{\lambda} q_0}
    &
    0
    &
    0
    &
    \cdots
    &
    0
    \\
    0
    &
    (q_0-1)^2k\phi
    &
    -e^{-1}\lambda e^{-k\frac{\eta}{\lambda} (q_0-1)}
    &
    0
    &
    \cdots
    &
    0
    \\
    \vdots
    &
    \ddots
    &
    \ddots
    &
    \ddots
    &
    \ddots
    &
    \vdots
    \\
    0
    &
    \cdots
    &
    0
    &
    0
    &
    k\phi
    &
    -e^{-1}\lambda e^{-k\frac{\eta}{\lambda}}
    \\
    0
    &
    \cdots
    &
    0
    &
    0
    &
    0
    &
    0
\end{pmatrix}.
\end{equation}
Moreover,
\begin{equation}
    u(t, x, q, s) = x + qs + \frac{1}{k}\log\big(\omega(t, q)\big)
\end{equation}
solves the HJB equation \eqref{eq:hjb}.
\end{theorem}
The proof of Theorem~\ref{thm:main} is given in Section~\ref{sec:proof-main}.

\begin{remark}
The function \(\omega\) in \eqref{eq:omega-matrix-representation} admits a recursive representation given by
\begin{equation}
\label{eq:omega-recursive-solution}
    \omega(t,q)
    =
    e^{-k\phi q^2(T-t)}e^{-k\alpha q^2}
    +
    e^{-1}\lambda e^{-k\frac{\eta}{\lambda} q}
    \int_t^T
        e^{-k\phi q^2(r-t)}
        \omega(r,q-1)\,dr,\qquad t\in [0,T],\ q\geq 1.
\end{equation}
\end{remark}

\begin{remark}
If one uses a CARA utility criterion, the optimal quote can be derived in an analogous way. In that case, letting $\gamma>0$ denote the trader’s CARA risk aversion and assuming $\lambda>\gamma\eta q_0$, the optimal quote is given by
\[
\delta^{\star}(t,q)
=
\frac{1}{k}
\log\bigg(
\frac{w_q^\gamma(t)}{w_{q-1}^\gamma(t)}
\bigg)
+
\frac{1}{\gamma}
\log\bigg(1+\frac{\gamma}{k}\bigg)
+
\frac{1}{\gamma}
\log\bigg(\frac{\lambda}{\lambda-\gamma\eta q}\bigg),
\qquad
t\in [0,T],\ q=1,\dots,q_0,
\]
where $w_q^\gamma(t)$ is the auxiliary function arising from the CARA utility. Therefore, passive impact shifts the optimal quote by the inventory-dependent term
$
\frac{1}{\gamma}
\log(\frac{\lambda}{\lambda-\gamma\eta q}).
$
For $\eta=0$, that is, no passive impact, one recovers the corresponding solution of \citet{Gueant2012Optimal} (see Theorem 1 therein).
\end{remark}

\section{Numerical illustrations}\label{sec:illustrations}

We now illustrate the behaviour of the optimal passive execution strategy for $m=k$. The parameters used for our figures and simulations are shown in Table \ref{tab:params}. We use a tick size of $\$0.01$, so that ticks are converted to dollars. In the numerical examples, one inventory unit corresponds to $1000$ shares. Thus one jump of \(N\) from \eqref{eq:execution-intensity-model} corresponds to the execution of \(1000\) shares, and \(q_0=20\) corresponds to \(20{,}000\) shares.

\begin{table}[H]
\begin{center}
\caption{Model parameter values.}\label{tab:params}
\vskip-0.2cm
\begin{tabular}{lcl}
\toprule
\textbf{Parameter} & \textbf{Value} & \textbf{Description}\\
\midrule
$T$          & 300   & trading horizon ($\mathrm{s}$) \\
$q_0$        & 20     & initial inventory (inventory units) \\
$S_0$        & 100 & initial midprice ($\$\cdot\mathrm{shares}^{-1}$) \\
$\sigma$     & 0.003  & price volatility ($\smash{\$\cdot\mathrm{shares}^{-1}\cdot\mathrm{s}^{-1/2}}$) \\
$\lambda$    & 1.4   & fill intensity ($\mathrm{inventory\ units}\cdot\mathrm{s}^{-1}$) \\
$k$     & 48   & fill intensity decay ($\$^{-1}$) \\
$\eta$      & 0.0028   & passive price impact coefficient ($\$\cdot\textrm{s}^{-1}$) \\
$\phi$       & 0.00001 & running inventory penalty ($\$\cdot\mathrm{s}^{-1}\cdot\mathrm{inventory\ units}^{-2}$) \\
$\alpha$     & 0.00001 & terminal inventory penalty ($\$\cdot\mathrm{inventory\ units}^{-2}$) \\
\bottomrule
\end{tabular}
\end{center}
\end{table}

Figure~\ref{fig:quotes-baseline} shows the optimal quote depth $\delta^\star(t,q)$ as a function of time. Each curve corresponds to a different inventory level $q$. The figure shows how the quote posted by the trader changes over the trading horizon and how this depends on the remaining inventory. Lighter colours indicate larger remaining inventories. We see that the optimal quote is closer to the midprice when the remaining inventory to liquidate is larger, corresponding to increased urgency due to risk aversion and terminal penalty. This effect reverses at the end of the trading period because the inventory will be marked to the impacted midprice, creating a situation where the trader would prefer to pay the inventory penalty of an additional unit than pay impact on the entire position.

\begin{figure}[H]
    \centering
    \includegraphics[width=0.64\textwidth]{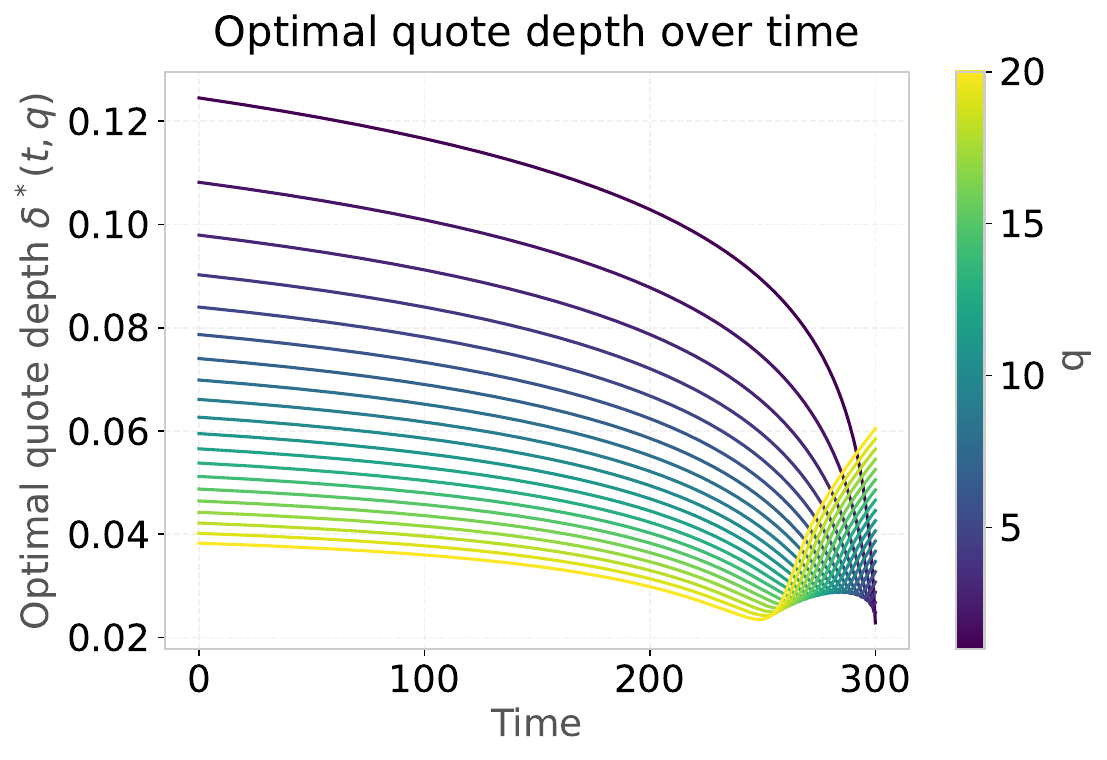}
    \caption{Optimal quote depth $\delta^\star(t,q)$ over time for different inventory levels $q$.}
    \label{fig:quotes-baseline}
\end{figure}

\begin{figure}[htb]
    \centering
    \includegraphics[width=0.7\textwidth]{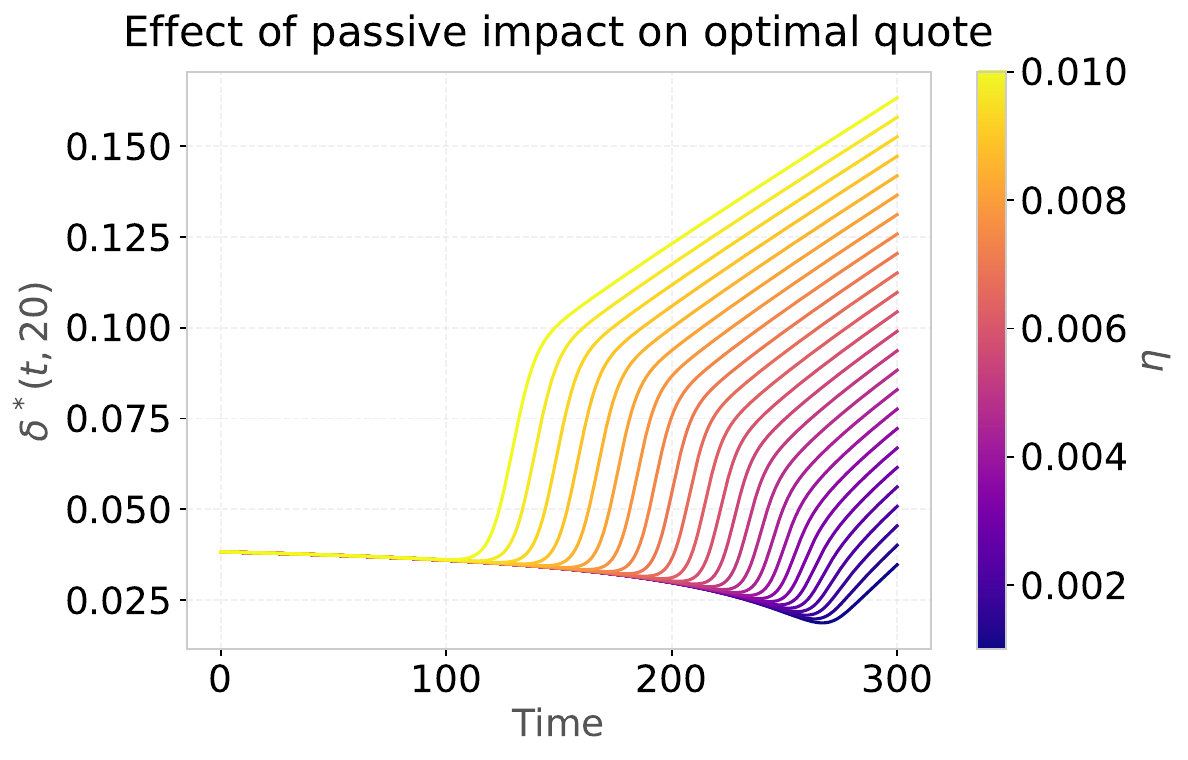}
    \caption{Effect of the passive impact parameter $\eta$ on the optimal quote depth $\delta^\star(t,20)$.}
    \label{fig:compare-eta}
\end{figure}

Figure~\ref{fig:compare-eta} compares the optimal quote depth $\delta^\star(t,20)$ across different values of the passive impact parameter $\eta$. Each curve corresponds to a different value of $\eta$, with darker curves representing larger passive impact. The figure shows that increasing $\eta$ shifts the optimal quote depth upwards. This means that when passive impact is stronger, the trader quotes less aggressively in order to reduce the price impact generated by executions. Additionally, we see that increased $\eta$ means that the trader enters earlier the state at which they prefer not to incur impact to protect the value of their position. Thus they begin to quote further from the midprice to compensate themselves for impact. Before that point, the impact has no bearing on their quotes, which are instead driven by risk aversion entirely.

\begin{figure}[htb]
    \centering
    \includegraphics[width=0.85\textwidth]{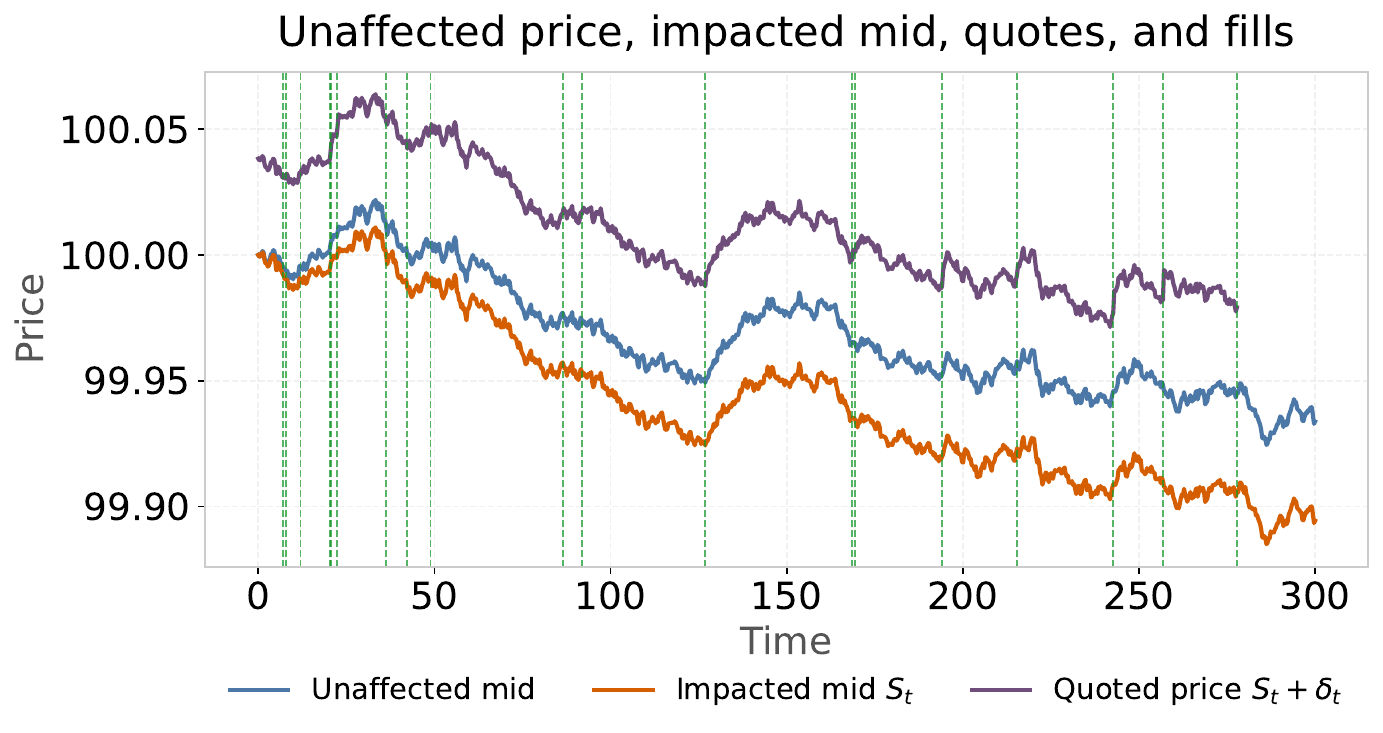}
    \caption{Simulated unaffected midprice (blue), impacted midprice (red), quoted price (purple), and fill times (green vertical lines) under the baseline strategy.}
    \label{fig:simulated-prices-baseline}
\end{figure}

Figure~\ref{fig:simulated-prices-baseline} shows a realisation of price paths using parameters from Table \ref{tab:params}. In particular the trader is executing $q_{0}=20$ units. The red line denotes the impacted midprice $S_t$, while the blue line denotes the unaffected price process. The purple line shows the quoted price $S_t+\delta_t$. The vertical green markers indicate the times at which passive limit orders are filled. At around the 275th second, the trader completes execution of their order and therefore their quoting stops.

\begin{table}[H]
\begin{center}
\caption{Performance metrics.}\label{tab:performance_metrics}
\vskip-0.2cm
\small
\setlength{\tabcolsep}{9pt}
\renewcommand{\arraystretch}{1.45}
\setlength{\dashlinedash}{0.4pt}
\setlength{\dashlinegap}{2pt}
\arrayrulecolor{gray!60}
\makebox[\textwidth][c]{%
\begin{tabular}{
    >{\raggedright\arraybackslash}p{0.18\textwidth}
    :
    >{\raggedright\arraybackslash}p{0.42\textwidth}
    :
    >{\centering\arraybackslash}p{0.28\textwidth}
}
\toprule
\textbf{Metric} & \textbf{Description} & \textbf{Expression} \\
\midrule
\arrayrulecolor{gray!60}

\textbf{Final inventory}
&
Number of remaining inventory units at the end of the trading period. Since the trader wishes to liquidate their entire position, smaller values are preferred.
&
$\displaystyle Q_T$
\\[0.2em]
\hdashline

\textbf{Trading time}
&
Time taken in seconds for the entire position to be liquidated, or $T$ if the position is never fully liquidated.
&
$\displaystyle \inf\{t\in[0,T] : Q_t = 0\} \wedge T$
\\[0.2em]
\hdashline

\textbf{Implementation shortfall}
&
Difference between the initial midprice and the average execution price. Positive implementation shortfall indicates a worse outcome, as the metaorder traded at a lower average price than the initial midprice.

Here, $\tau_n$ is the time of the $n$-th fill, given by $\inf\{\tau\in[0,T] : N_\tau = n\}$. Recall that $(N_{t})_{t\in [0,T]}$ is the counting process representing cumulative limit order fills.
&
$\displaystyle
S_0 - \frac{1}{N_T}\sum_{n=1}^{N_T}
(S_{\tau_n} + \delta_{\tau_n})
$
\\[0.2em]
\hdashline

\textbf{P\&L}
&
Net trading profit in dollars from executing the trade. This is the interior of the value function \eqref{eq:optimisation-problem}, without the inventory penalty terms and the value of the initial inventory marked to the midprice.
&
$\displaystyle X_T + Q_T S_T - q_0 S_0$
\\[0.2em]

\arrayrulecolor{black}
\bottomrule
\end{tabular}%
}
\end{center}
\end{table}

Table \ref{tab:eta_summary} reports means (with 95\% confidence intervals) of several metrics for a Monte Carlo simulation of the model for different values of the permanent price impact parameter $\eta$ defined in \eqref{eq:S}. We simulate $1000$ paths, with $1000$ time-steps each. The metrics we use are final inventory, trading time, implementation shortfall, and P\&L. Each is defined in Table \ref{tab:performance_metrics}.

\begin{table}[H]
\begin{center}
\caption{Summary statistics by $\eta$.}\label{tab:eta_summary}
\vskip-0.2cm
\small
\setlength{\tabcolsep}{6pt}
\renewcommand{\arraystretch}{1.15}
\setlength{\dashlinedash}{0.4pt}
\setlength{\dashlinegap}{2pt}
\arrayrulecolor{gray!60}
\makebox[\textwidth][c]{%
\begin{tabular}{l:c:c:c:c}
\toprule
\textbf{$\eta$}
& \begin{tabular}{@{}c@{}}\textbf{Final inventory}\\ \textbf{(thousand shares)}\end{tabular}
& \begin{tabular}{@{}c@{}}\textbf{P\&L}\\ \textbf{(thousand \$)}\end{tabular}
& \begin{tabular}{@{}c@{}}\textbf{Implementation shortfall}\\ \textbf{(thousand \$)}\end{tabular}
& \begin{tabular}{@{}c@{}}\textbf{Trading time}\\ \textbf{(s)}\end{tabular} \\
\midrule
\arrayrulecolor{gray!60}

$\mathbf{0.0}$
& \begin{tabular}{@{}c@{}}$1.003$\\$(0.932,\ 1.074)$\end{tabular}
& \begin{tabular}{@{}c@{}}$0.320$\\$(0.301,\ 0.339)$\end{tabular}
& \begin{tabular}{@{}c@{}}$-0.035$\\$(-0.037,\ -0.034)$\end{tabular}
& \begin{tabular}{@{}c@{}}$284.918$\\$(283.257,\ 286.578)$\end{tabular}
\\
\hdashline

$\mathbf{0.005}$
& \begin{tabular}{@{}c@{}}$1.928$\\$(1.816,\ 2.040)$\end{tabular}
& \begin{tabular}{@{}c@{}}$0.071$\\$(0.049,\ 0.093)$\end{tabular}
& \begin{tabular}{@{}c@{}}$-0.019$\\$(-0.021,\ -0.017)$\end{tabular}
& \begin{tabular}{@{}c@{}}$290.153$\\$(288.674,\ 291.631)$\end{tabular}
\\
\hdashline

$\mathbf{0.01}$
& \begin{tabular}{@{}c@{}}$6.367$\\$(6.195,\ 6.538)$\end{tabular}
& \begin{tabular}{@{}c@{}}$-0.011$\\$(-0.040,\ 0.019)$\end{tabular}
& \begin{tabular}{@{}c@{}}$-0.061$\\$(-0.063,\ -0.058)$\end{tabular}
& \begin{tabular}{@{}c@{}}$298.314$\\$(297.617,\ 299.012)$\end{tabular}
\\

\arrayrulecolor{black}
\bottomrule
\end{tabular}%
}
\end{center}
\end{table}
We observe in the first column of Table \ref{tab:eta_summary} that when there is more impact for posting, the trader on average liquidates much less of their position over the trading period due to more conservative quoting. Nevertheless, they also have lower P\&L, as seen in the second column, but not worse implementation shortfall,  as seen in the third column. This is because the higher distances from the midprice compensate for the increased impact. This more conservative quoting also leads to higher trading times on average, as seen in the fourth column. 

\begin{figure}[H]
    \centering
    \includegraphics[width=0.85\textwidth]{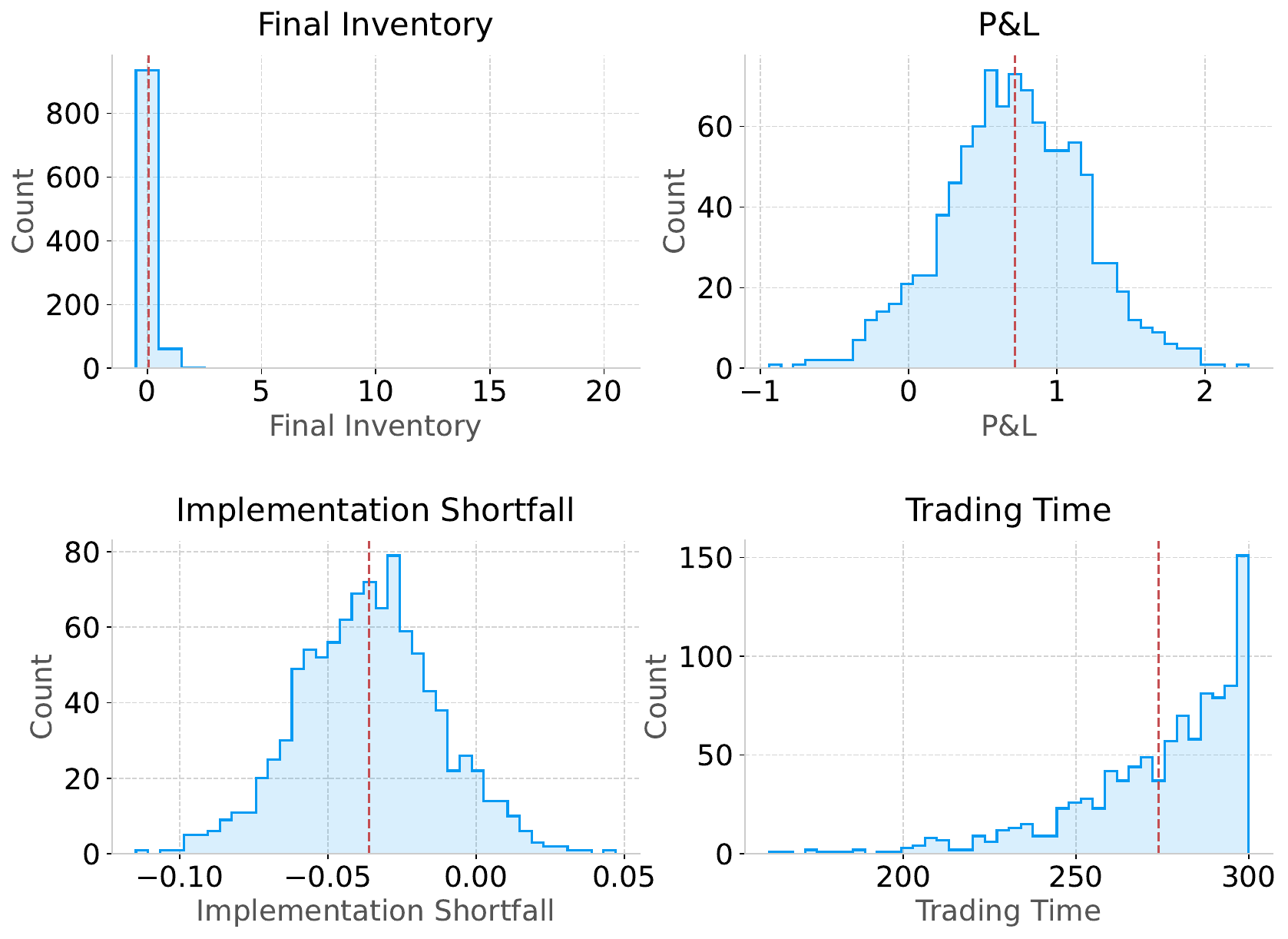}
    \caption{Monte Carlo diagnostics for final inventory, P\&L, implementation shortfall, and trading time for the parameters in Table \ref{tab:params}. Dashed red lines indicate sample means.}
    \label{fig:mc-diagnostics-baseline}
\end{figure}

Finally, Figure \ref{fig:mc-diagnostics-baseline} displays histograms for the parameter specification from Table~\ref{tab:params}. The four panels show the distributions of final inventory, P\&L, implementation shortfall, and total trading time. In each panel, the dashed red line denotes the sample means across $1000$ simulated paths with $1000$ time-steps each.

\section{Empirical evidence} \label{sec:empirics}

In this section, we estimate two different decays with respect to the distance from the midprice $\delta$ on both public US stock data and LSEG FX market data: the fill intensity decay $k$ from \eqref{eq:execution-intensity} and the MLOFI price impact decay $\ell$ from \eqref{eq:mlofi-decay}. Recall that the passive impact decay with respect to $\delta$ is given by their sum $m=k+\ell$ introduced in \eqref{eq:passive-impact}.

\subsection{NASDAQ stock data}

We estimate the decays $k$ and $\ell$ (and thereby $m$) for six US stocks traded on the NASDAQ, ranging from small-tick to large-tick stocks.

\subsubsection{Data} 

We use the same LOBSTER data set as \citet{xu2018multi} in their study of MLOFI. This choice ensures that our estimated fill intensity decay can be compared directly with the decay of their published MLOFI regression coefficients.

LOBSTER provides an event-by-event record of the NASDAQ limit order book, including visible limit order submissions, executions, cancellations, and deletions, with precise timestamps, prices, directions, and order sizes. NASDAQ is a price-time priority market with a tick size of $\$0.01$ and a lot size of one share. Following \cite{xu2018multi}, we consider six stocks with different relative tick sizes listed in increasing order: AMZN, TSLA, NFLX, ORCL, CSCO, MU. The sample covers all 252 trading days from 4 January 2016 to 30 December 2016. Weekends and public holidays are excluded, and only the continuous trading period from 10:00 to 15:30 is used, thereby removing the first and last 30 minutes of each trading day.

\subsubsection{Estimating the fill intensity decay} \label{sec:nasdaq_k_estimation}

We estimate fill intensities from the LOBSTER message and orderbook files using the same time period and intraday window as \cite{xu2018multi}. For each event, we reconstruct the pre-event best bid, best ask, midprice, and spread. We track visible limit orders from submission until their first execution, cancellation, or deletion. For an order with limit price \(p\), its current distance from the midprice in $\textrm{ticks}$ is defined by
\[
    \delta=\frac{|\textrm{mid}-p|}{\text{tick size}}.
\]
For each distance bin \(\delta\), we estimate the fill intensity in $\textrm{shares}\cdot\mathrm{s}^{-1}$ by
\be\label{eq:estimate-k}
    \hat{\Lambda}(\delta)
    =
    \frac{N_{\mathrm{fill}}(\delta)}
    {T_{\mathrm{exp}}(\delta)},
\ee
where \(N_{\mathrm{fill}}(\delta)\) is the total volume of executions occurring at current distance \(\delta\) in $\textrm{shares}$, and \(T_{\mathrm{exp}}(\delta)\) is the total exposure time accumulated by tracked visible orders while their current distance is \(\delta\) in $\mathrm{s}$. Both quantities are summed across all trading days before computing the ratio. Then, to estimate $\lambda$ and $k$ in \eqref{eq:execution-intensity}, we use the estimated fill intensity $\hat{\Lambda}(\delta)$ to fit
\be\label{eq:fit-k}
    \log {\hat\Lambda}(\delta)
    =
    \log \lambda - k\delta.
\ee
The intercept gives \(\log\hat{\lambda}\), while the negative slope gives \(\hat{k}\).

\subsubsection{Estimating the MLOFI price impact decay} \label{sec:nasdaq_l_estimation}

To estimate the decay of the MLOFI price impact across the order book, we use the MLOFI regression coefficients reported in \cite{xu2018multi}. Recalling \eqref{eq:DeltaS} and \eqref{eq:MLOFI}, the regression in \cite{xu2018multi} relates contemporaneous price changes to MLOFI at the first $\bar d$ occupied price levels,
\begin{equation}\label{eq:MLOFI-regression}
    \Delta S_{t,t+h}
    =
    \sum_{d=1}^{\bar d}
    \beta^d
    \mathrm{MLOFI}^d_{t,t+h}
    +
    \varepsilon_{t,t+h} .
\end{equation}
The authors of \cite{xu2018multi} estimate \eqref{eq:MLOFI-regression} for \(\bar d=10\) using both OLS and ridge regression. We use the ridge coefficients reported in Table 8 of \cite{xu2018multi}, because the MLOFI regressors are strongly collinear across neighbouring book levels and ridge regression produces more stable and statistically significant estimates than OLS.  

We interpret the coefficients \(\beta^d\) as level-specific price impact coefficients: a positive net MLOFI at level \(d\) is associated with a contemporaneous movement of the midprice in the same direction. Since our fill intensity estimates are expressed as a function of distance from the midprice rather than book level, we first convert each level \(d\) into an average distance \(\delta_d\). The distance of level 1 is taken to be the average half-spread, and deeper levels are obtained by adding the average gap between adjacent occupied bid and ask levels. We then fit an exponential decay of MLOFI price impact $\beta(\delta)=\xi e^{-\ell \delta}$ by
\be\label{eq:fit-l}
    \log \beta^d = \log \xi - \ell \delta_d,
    \qquad d=1,\ldots,\bar d.
\ee
The intercept gives \(\log \hat{\xi}\) and the negative slope gives \(\hat{\ell}\).

\citet{Cont2014Price} report price impact coefficients in ticks per 100 shares of OFI. In contrast, \citet{xu2018multi} report price impact coefficients in ticks per \(10000\) shares of MLOFI. We therefore divide the latter coefficients by \(100\), yielding coefficients in ticks per 100 shares of MLOFI, consistent with the convention of \cite{Cont2014Price}.

\subsubsection{NASDAQ empirical decay estimates} \label{sec:nasdaq_empirical_decay}

We now turn to the estimates of the two decay parameters. The first is the fill intensity decay \(k\), obtained from the exponential fit \eqref{eq:fit-k}.
The second is the MLOFI price impact decay \(\ell\), obtained from the exponential fit \eqref{eq:fit-l}.
We recall that for our model from Section~\ref{sec:model} the relevant combined passive impact decay rate is $m = k+\ell$.

\begin{figure}[htb]
    \centering
    \includegraphics[width=\textwidth]{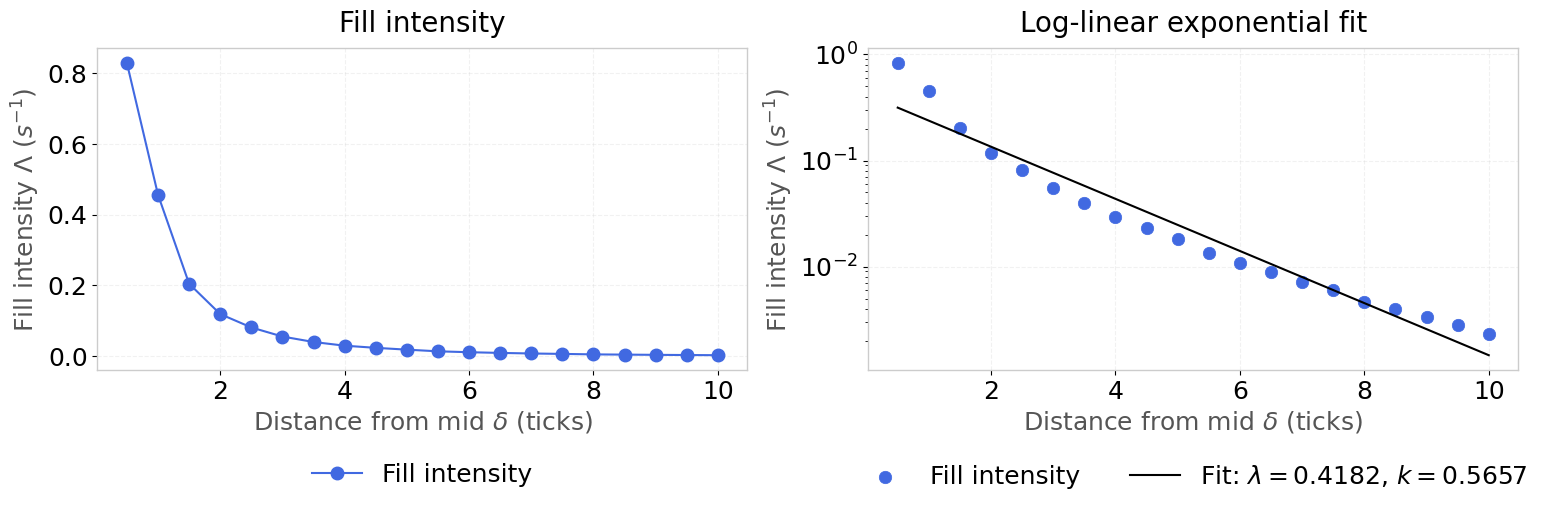}
    \caption{Fill intensity estimation for TSLA. The left panel shows the empirical one-second fill probability as a function of the current distance from the midprice $\delta$. The right panel shows the fill intensity on a logarithmic scale, together with the exponential fit \(\Lambda(\delta)=\lambda e^{-k\delta}\).}
    \label{fig:tsla-fill-decay}
\end{figure}

Figure~\ref{fig:tsla-fill-decay} shows the estimation of \(k\) for TSLA. The left panel reports the empirical fill intensity \eqref{eq:estimate-k} as a function of the current distance from the midprice, while the right panel shows the corresponding exponential fit \eqref{eq:fit-k}. The fill intensity decays rapidly and approximately exponentially with distance, giving a well-defined estimate of \(k\). Figure~\ref{fig:tsla-impact-decay} shows the corresponding estimation of \(\ell\) for TSLA. The left panel shows the MLOFI price impact coefficients from Table 8 of \cite{xu2018multi}, whereas the right panel shows the corresponding exponential fit \eqref{eq:fit-l}. The MLOFI price impact coefficients decay much more slowly than the fill intensity, which is visible directly from the flatter fit.

\begin{figure}[htb]
    \centering
    \includegraphics[width=\textwidth]{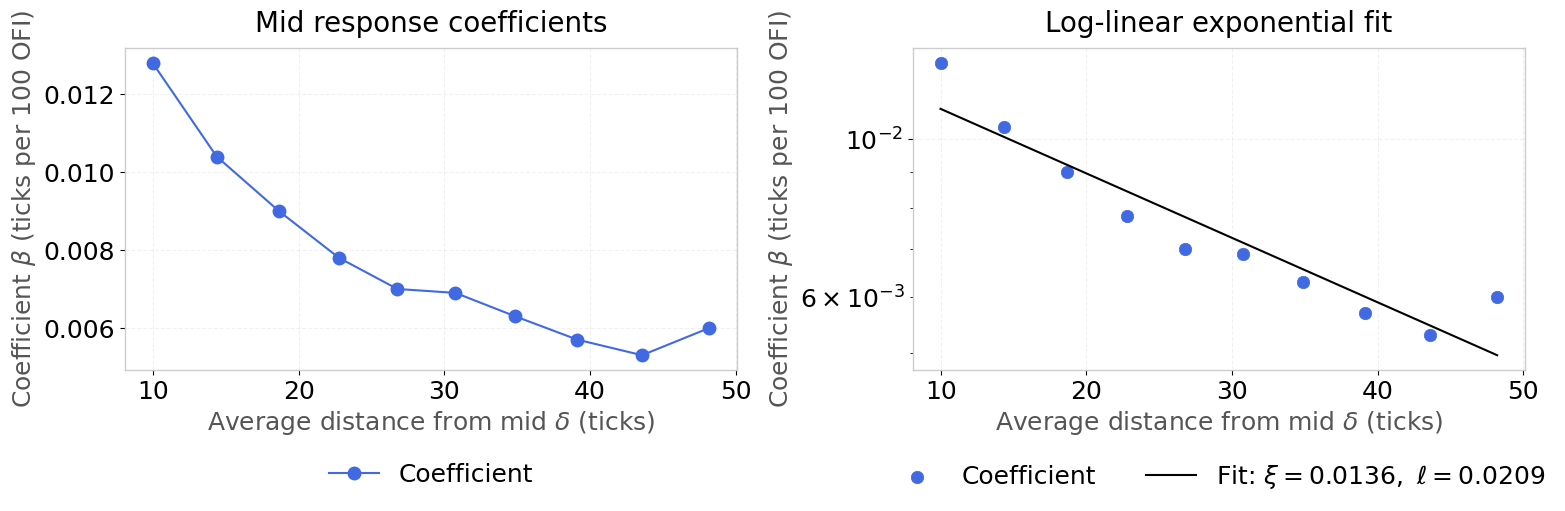}
    \caption{Passive impact estimation for TSLA. The left panel shows the MLOFI price impact coefficients as a function of the average distance from the midprice $\delta$. The right panel shows the same coefficients on a logarithmic scale, together with the exponential fit \(\beta(\delta)=\xi e^{-\ell\delta}\).}
    \label{fig:tsla-impact-decay}
\end{figure}

Table~\ref{tab:tsla-parameter-estimates} summarises the estimated parameters for all six stocks. The main empirical finding is that the MLOFI price impact decay is consistently smaller than the fill intensity decay, that is,
\begin{equation}\label{eq:l<k}
    \hat{\ell} \ll \hat{k}
\end{equation}
for all six stocks in the sample. Here, the confidence intervals are roughly one to two orders of magnitude smaller than the corresponding parameter values, but are not reported here. The fill intensity estimates produce decay rates $k$ in the range \(0.48\)--\(3.72\). By contrast, the MLOFI price impact decay is much smaller, typically close to zero and generally between approximately \(0.01\) and \(0.02\) for small-tick stocks and between approximately \(-0.10\) and \(-0.06\) for large-tick stocks. This indicates that fill probabilities decrease rapidly as a limit order is placed farther from the midprice as the order becomes more unlikely to be filled, whereas the associated MLOFI price impact coefficients are much more persistent across book levels.  This may be explained by the fact that limit orders deeper in the book still reveal information, thereby moving the price without being filled. In particular, the passive impact decay parameter $m=k+\ell$ satisfies \(m \approx k\), consistent with our model from Section~\ref{sec:model}.

\begin{table}[htb]
\begin{center}
\caption{Estimated fill intensity and passive impact parameters for US stocks.}
\label{tab:tsla-parameter-estimates}
\vskip-0.2cm
\small
\setlength{\tabcolsep}{6pt}
\renewcommand{\arraystretch}{1.15}
\setlength{\dashlinedash}{0.4pt}
\setlength{\dashlinegap}{2pt}
\arrayrulecolor{gray!60}
\makebox[\textwidth][c]{%
\begin{tabular}{c:c@{\;}l:c@{\;}l:c@{\;}l:c@{\;}l:c@{\;}l:c@{\;}l}
\toprule
& \multicolumn{2}{c}{\textbf{AMZN}} & \multicolumn{2}{c}{\textbf{TSLA}} & \multicolumn{2}{c}{\textbf{NFLX}} & \multicolumn{2}{c}{\textbf{ORCL}} & \multicolumn{2}{c}{\textbf{CSCO}} & \multicolumn{2}{c}{\textbf{MU}} \\
\midrule
\arrayrulecolor{gray!60}
$\textbf{Tick size}$                 & {} & $0.01$ & {} &  $0.01$  & {} & $0.01$ & {} & $0.01$ & {} & $0.01$ & {} & $0.01$ \\
\hdashline
\textbf{Spread (ticks)}  & {} & $36.7$  & {} & $19.2$ & {} & $4.0$     & {} & $1.2$   & {} & $1.1$   & {} & $1.1$ \\
\hdashline
$\hat{\boldsymbol{k}}$       & {} & $0.4807$ & {} & $0.5657$ & {} & $0.8584$ & {} & $1.5781$ & {} & $3.7247$ & {} & $1.6858$ \\
\hdashline
$\hat{\boldsymbol{\ell}}$    & {} & $0.0167$ & {} & $0.0209$ & {} & $0.0150$ & $-$ & $0.0650$ & $-$ & $0.1015$ & $-$ & $0.0637$ \\
\hdashline
$\hat{\boldsymbol{m}}$       & {} & $0.4974$ & {} & $0.5866$ & {} & $0.8734$ & {} & $1.5131$ & {} & $3.6232$ & {} & $1.6221$ \\
\arrayrulecolor{black}
\bottomrule
\end{tabular}%
}
\end{center}
\end{table}

The difference between the estimated decays \eqref{eq:l<k} is especially clear for the small-tick stocks AMZN, NFLX, and TSLA, for which the MLOFI coefficients display a visible decay with distance from the midprice. For the larger-tick stocks ORCL, CSCO, and MU, the coefficients reported in Table 8 of \cite{xu2018multi} do not exhibit a clean monotone decay across levels. The fitted \(\ell\) is slightly negative, but close to zero, for these stocks. The estimates of $\ell$ for the large-tick stocks should therefore be interpreted with particular caution, even though they are statistically significant. MLOFI price impact may depend more directly on the volume queued ahead of an order than on its distance from the midprice. While these quantities may be reasonable proxies for one another in small-tick stocks, the coarser price grid of large-tick stocks may prevent distance from providing a sufficiently informative substitute for queue volume. The slightly negative estimates may thus partly reflect this limitation of the model specification rather than genuine amplification of passive impact across book levels for large-tick stocks. 


\subsection{FX market data}

Having studied the empirical decays $k$ and $\ell$ in the equity market, we now estimate these parameters for five currency pairs with various tick sizes. As in the equity sample, we typically find that $\ell<k$, however, $\ell$ is generally larger relative to $k$ in the currency market.

\subsubsection{Data and estimation}

We estimate the decays $k$ and $\ell$ using LSEG quote and trade data for several FX pairs of varying liquidity with different tick sizes: USDMXN, GBPUSD, AUDUSD, USDTHB, USDSGD. The quote data contain the first $\bar d=10$ price levels on each side of the book. The FX market is a $24$-hour market with well-established intraday seasonality; we focus on the liquid window from $06:00$ to $18:00$ UTC on business days only, excluding weekends and holidays. Trade volume and quote sizes are measured in millions of base currency. The L2 order book data is conflated with at least $5$ ms between snapshots. Therefore, contrary to the equity market, detailed microstructure events are not available. The dataset included $100$ business days of history from the beginning of $2026$.

Both the fill intensity and MLOFI price impact estimates are computed using 10-second buckets. For each bucket, we record the final midprice and the average spread in ticks. We estimate the parameters $k$ and $\ell$ using the same log-linear regression procedure as in Sections \ref{sec:nasdaq_k_estimation} and \ref{sec:nasdaq_l_estimation} respectively, with separate treatment of buy and sell orders. In particular, the estimation of \(\ell\) is based on an MLOFI regression of the form \eqref{eq:MLOFI-regression}, analogous to the regression used for US stocks in \cite{xu2018multi}. For FX, the MLOFI variables are constructed from quote updates following the Cont--Kukanov--Stoikov convention and are then aggregated by book level within each bucket. The default response variable is the contemporaneous close-to-close bucket midprice change in ticks. The reported distance-to-mid impact decay is fitted to the positive unconstrained per-level MLOFI coefficients in log scale. For fill intensity, the implementation uses synthetic passive fills from public quote and trade data. A passive quote at distance $d$ from the midprice is considered eligible only when it is passive relative to the prevailing spread, and fills are proxied by next-bucket aggressor trades that reach or cross the hypothetical quote. The resulting exposure-style estimate fits $\lambda(d)=\lambda_0 e^{-kd}$ over the selected log-linear region and is the closest analogue of the LOBSTER fill intensity estimator.

\subsubsection{FX empirical decay estimates}

As in Section \ref{sec:nasdaq_empirical_decay}, we now examine the estimates of the fill intensity decay $k$ and the MLOFI price impact decay $\ell$, and again recall that the passive impact decay is $m=k+\ell$.

Figure \ref{fig:fx-fill-decay} shows the estimation of $k$ for GBPUSD. The left panel displays the empirical 10-second fill intensity as a function of the distance in ticks from the midprice, while the right panel reports the same data on a log-linear scale with the fitted exponential. Both buy and sell intensities can be seen to follow the same shape, with some noise deeper in the book. As for TSLA in Figure \ref{fig:tsla-fill-decay}, the fill intensities decay very rapidly, since limit orders placed farther from the midprice become increasingly unlikely to be reached and executed. Figure \ref{fig:fx-impact-decay} shows the estimation of $\ell$ for GBPUSD, and we see that as for TSLA in Figure \ref{fig:tsla-impact-decay} the impact decays more slowly than the fill probability (that is, $\hat\ell<\hat k$).

\begin{figure}[H]
    \centering
    \includegraphics[width=\textwidth]{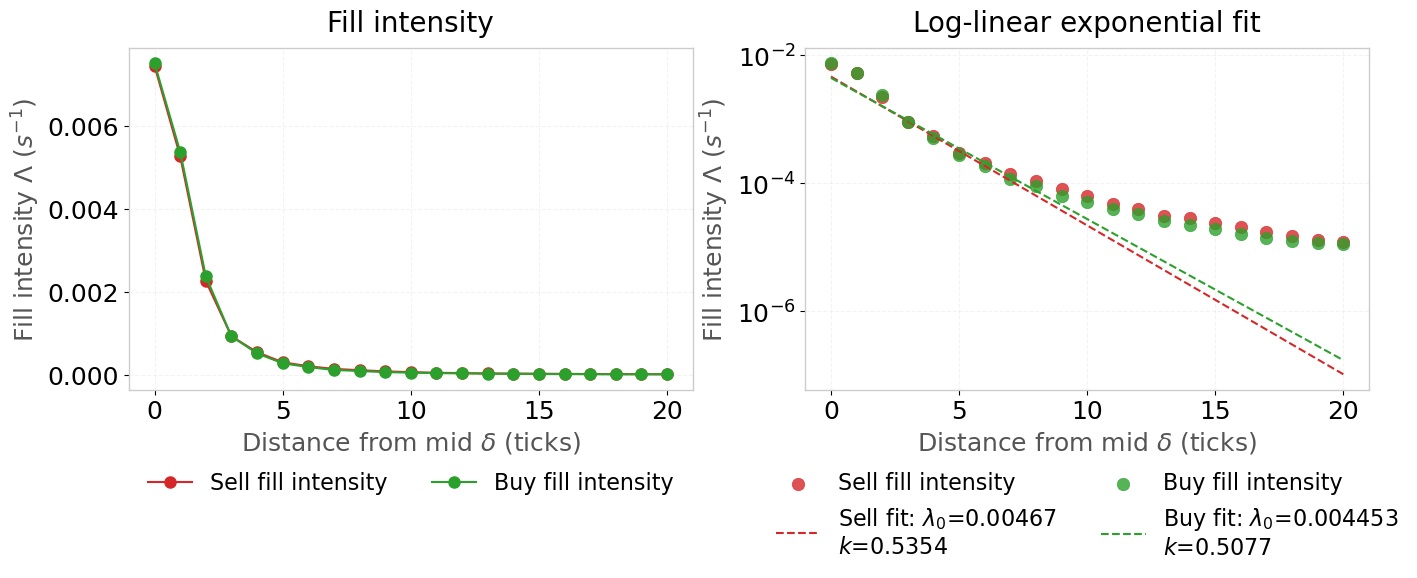}

    \caption{Fill intensity estimation for GBPUSD. The left panel shows the empirical proxied $10$-second fill intensity as a function of the current distance from the midprice $\delta$. The right panel shows the fill intensity on a logarithmic scale, together with the exponential fit \(\Lambda(\delta)=\lambda e^{-k\delta}\).}
    \label{fig:fx-fill-decay}
\end{figure}

\begin{figure}[H]
    \centering
    \includegraphics[width=\textwidth]{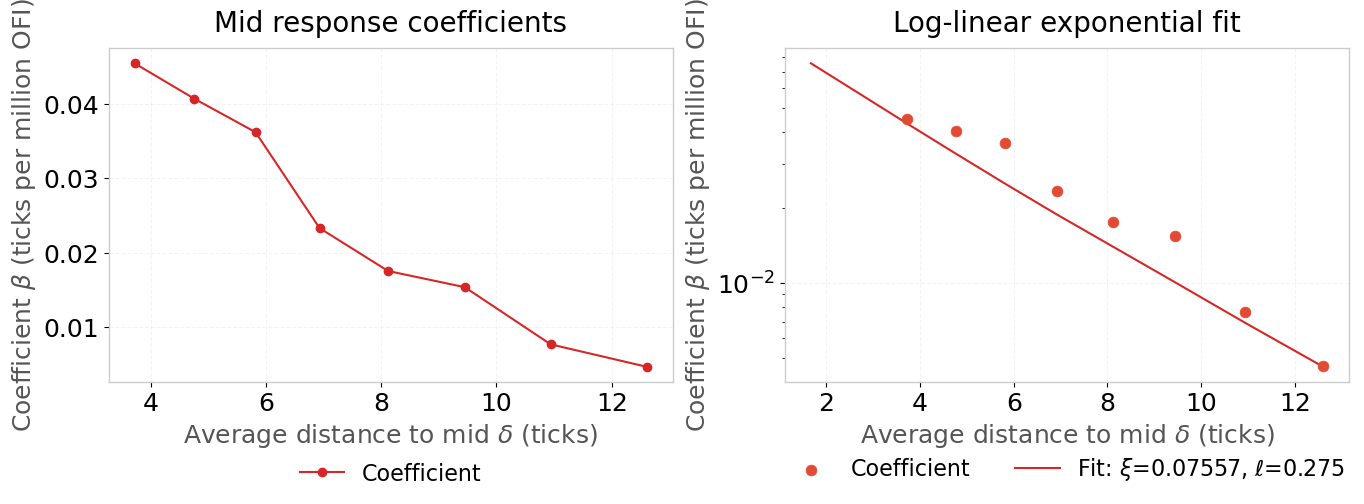}

    \caption{Passive impact estimation for GBPUSD. The left panel shows the MLOFI price impact coefficients as a function of the average distance from the midprice $\delta$. The right panel shows the same coefficients on a logarithmic scale, together with the exponential fit $\xi e^{-\ell\delta}$.}
    \label{fig:fx-impact-decay}
\end{figure}

Table \ref{tab:fx-parameter-estimates} reports the estimated parameters, along with tick sizes and average spreads, for the five currency pairs. These pairs span both major and emerging market currencies, as well as different tick sizes. For FX the MLOFI price impact decay $\ell$ relative to the fill intensity decay $k$ is found to be larger than for equities in Table \ref{tab:tsla-parameter-estimates}. The finding that $m \approx k$ for equities that justified our model is therefore not replicated as strongly in the FX market, especially for some pairs. This comparison should be interpreted as a public-data, venue-level result: the estimates use LSEG L2 quote and trade data only, not dealer OTC flow. One possible explanation is that public venue-level information captures only part of the total FX liquidity and information set, whereas the public equity order book is a more central source of information for the corresponding stocks. This evidence indicates that allowing the passive impact and fill intensity to decay at distinct rates, $m\neq k$, is important for capturing market-specific differences between equities and FX. Accordingly, Section~\ref{sec:different} develops and analyses this more general specification, which broadens the applicability of our framework beyond the benchmark case $m\approx k$.

\begin{table}[H]
\begin{center}
\caption{Estimated fill intensity and passive impact parameters for FX pairs.}
\label{tab:fx-parameter-estimates}
\vskip-0.2cm
\small
\setlength{\tabcolsep}{6pt}
\renewcommand{\arraystretch}{1.15}
\setlength{\dashlinedash}{0.4pt}
\setlength{\dashlinegap}{2pt}
\arrayrulecolor{gray!60}
\makebox[\textwidth][c]{%
\begin{tabular}{c:c@{\;}l:c@{\;}l:c@{\;}l:c@{\;}l:c@{\;}l}
\toprule
& \multicolumn{2}{c}{\textbf{USDMXN}}
& \multicolumn{2}{c}{\textbf{GBPUSD}}
& \multicolumn{2}{c}{\textbf{AUDUSD}}
& \multicolumn{2}{c}{\textbf{USDTHB}}
& \multicolumn{2}{c}{\textbf{USDSGD}} \\
\midrule
\arrayrulecolor{gray!60}
$\textbf{Tick size}$      & {} & $0.001$ & {} & $0.00005$ & {} & $0.00005$ & {} & $0.005$ & {} & $0.0001$ \\
\hdashline
\textbf{Spread (ticks)}   & {} & $4.4$ & {} & $2.7$ & {} & $2.3$ & {} & $2.1$ & {} & $1.5$ \\
\hdashline
$\hat{\boldsymbol{k}}$    & {} & $0.52$ & {} & $0.49$ & {} & $0.54$ & {} & $0.62$ & {} & $0.61$ \\
\hdashline
$\hat{\boldsymbol{\ell}}$ & {} & $0.23$ & {} & $0.27$ & {} & $0.13$ & {} & $0.33$ & {} & $0.26$ \\
\hdashline
$\hat{\boldsymbol{m}}$    & {} & $0.75$ & {} & $0.76$ & {} & $0.67$ & {} & $0.95$ & {} & $0.87$ \\
\arrayrulecolor{black}
\bottomrule
\end{tabular}%
}
\end{center}
\end{table}

\section{Extensions} \label{sec:extensions}

In this section, we consider three extensions of the model introduced in Section~\ref{sec:model}. The first allows for different decay rates \(m\neq k\), the second addresses more general transient passive impact instead of purely permanent passive impact, and the third studies target execution schedules, deviations from which are penalised. The results of this section are proved in Section~\ref{sec:proofs-extensions}.

\subsection{Different exponential decay rates} \label{sec:different}

The explicit solution in Theorem~\ref{thm:main} relies on the fact that the fill intensity and the passive impact term have the same exponential decay in the quote distance $\delta$. In light of our empirical findings from Table~\ref{tab:fx-parameter-estimates}, it is natural to ask what happens if this assumption is relaxed. 

Suppose therefore that executions still arrive with intensity $\lambda e^{-k\delta}$ as in \eqref{eq:execution-intensity-model}, but that the passive impact term $\eta e^{-m \delta}$ in \eqref{eq:S} decays at a possibly different exponential rate, i.e, suppose \(m\neq k\). The associated HJB equation \eqref{eq:hjb} is then modified only at the impact term, that is,
\be
\begin{aligned}\label{eq:hjb-diffdecay}
    0
    =
    &\frac{\partial  u}{\partial t}(t, x, q, s)
    +\frac12\sigma^2\frac{\partial^{2} u}{\partial s^{2}}(t, x, q, s)
    -\phi q^2\\
    &+
    \mathbf 1_{\{q>0\}}
    \sup_{\delta}
    \bigg\{
        \lambda e^{-k\delta}
        \big(
             u(t,x+s+\delta,q-1,s)- u(t,x,q,s)
        \big)
        -
        \eta e^{-m\delta}\frac{\partial u}{\partial s}(t,x,q,s)
    \bigg\},
\end{aligned}
\ee
with the same boundary and terminal conditions
\be
     u(t,x,0,s) = x, 
     \qquad
     u(T,x,q,s) = x + qs - \alpha q^{2}.
\ee
Let \(W_0\) denote the principal real branch of the Lambert \(W\) function, defined by
\begin{equation}\label{eq:LambertW}
    W_0(z)e^{W_0(z)}=z,
    \qquad z\in[-e^{-1},\infty),
\end{equation}
with \(W_0(z)\geq -1\).
We also define the auxiliary function $\Phi_q:\R\to\R$ by
\begin{equation}
\label{eq:Phi-q-def}
\Phi_q(z)
:=
W_0\bigg(
    \frac{m-k}{k}\frac{\eta}{\lambda} q m
    \exp\bigg(
        -\frac{m-k}{k}
        (1+kz)
    \bigg)
\bigg).
\end{equation}

\begin{theorem}\label{thm:decay}
Assume either that $m>k$ or that $k-\eps< m<k$ for a sufficiently small $\eps>0$. For $q=1,...,q_{0}$, the optimal quote is given by
\begin{equation}
\label{eq:optimal-quote-diffdecay}
\delta^\star(t,q)
    =
    \frac1k+\theta(t,q)-\theta(t,q-1)
    +
    \frac{1}{m-k}
    \Phi_q\big(\theta(t,q)-\theta(t,q-1)\big)
\end{equation}
where \(\theta(t,q)\) solves the nonlinear triangular system
\begin{equation}
\label{eq:theta-diffdecay}
\begin{aligned}
    \frac{\partial\theta}{\partial t}(t,q)
    =
    \phi q^2 
    -&
    \lambda \bigg(
        \frac1k
        +
        \frac{1}{m}
        \Phi_q\big(\theta(t,q)-\theta(t,q-1)\big)
    \bigg)
     \\
    &\times
    \exp\bigg(
        -1
        -k\big(\theta(t,q)-\theta(t,q-1)\big)
        -
        \frac{k}{m-k}
        \Phi_q\big(\theta(t,q)-\theta(t,q-1)\big)
    \bigg)
\end{aligned}
\end{equation}
with boundary and terminal conditions
\be\label{eq:boundary-terminal-diffdecay}
    \theta(t,0) = 0,
    \qquad 
    \theta(T,q)=-\alpha q^2.
\ee
\end{theorem}
Theorem \ref{thm:decay} proves that different exponential decay rates still lead to a semi-explicit optimal quote \eqref{eq:optimal-quote-diffdecay} through the Lambert \(W\) function, but the value functions now solve the nonlinear triangular ODE system \eqref{eq:theta-diffdecay}--\eqref{eq:boundary-terminal-diffdecay}. Hence the closed-form solution for \(\theta\), available in the equal-decay case \(m=k\), is lost.

\begin{remark}
The case \(m=k\) is recovered from Theorem~\ref{thm:decay} by taking the limit \(m\to k\). Indeed, using \eqref{eq:Phi-q-def} and the fact that $W_0'(0)=1$, the last term on the right-hand side of \eqref{eq:optimal-quote-diffdecay} converges to \(\frac{\eta}{\lambda} q\), and one obtains
\[
    \delta^\star(t,q)
    =
    \frac1k
    +
    \theta(t,q)-\theta(t,q-1)
    +
    \frac{\eta}{\lambda} q,
\]
which is precisely the optimal quote in the main model.
\end{remark}

\begin{figure}[htb]
    \centering
    \includegraphics[width=0.7\textwidth]{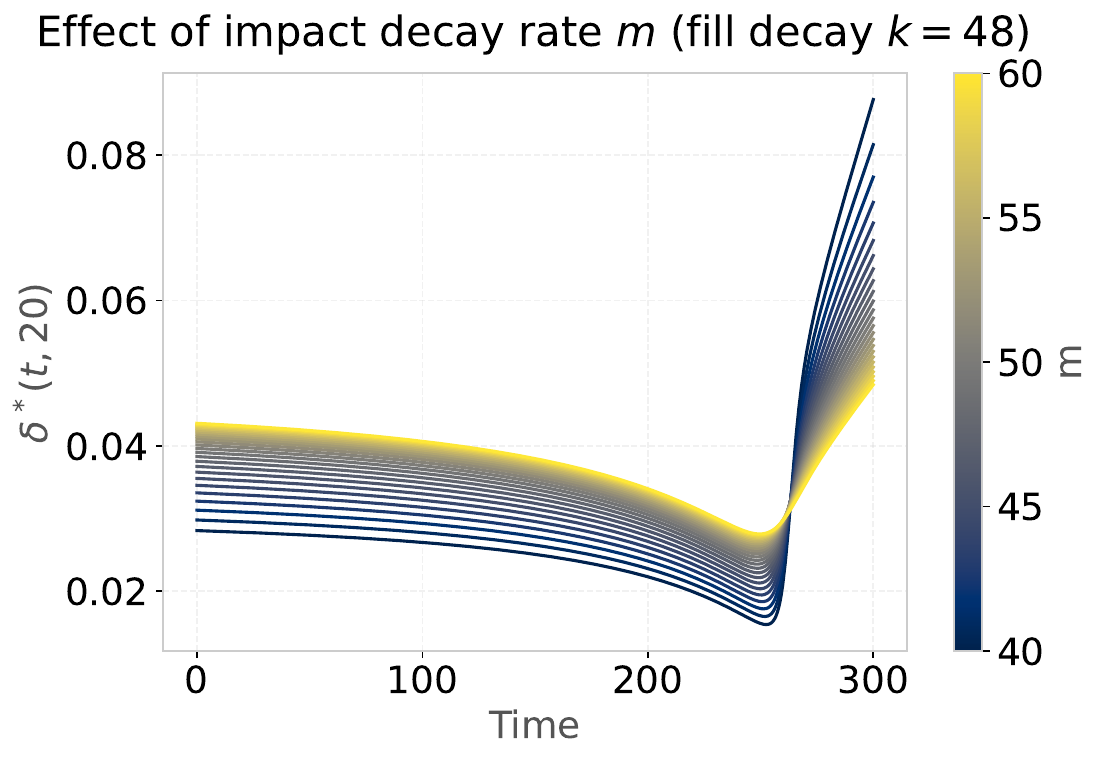}
    \caption{Effect of the passive impact decay parameter $m$ on the optimal quote depth $\delta^{\star}(t, 20)$.}
    \label{fig:compare_m}
\end{figure}

Figure \ref{fig:compare_m} compares the optimal quote depth $\delta^{\star}(t, 20)$ across different values of the passive impact decay parameter $m$ using the parameters from Table~\ref{tab:params}, with darker blue colours indicating slower decay as quote depth increases. We see that as the impact decays more quickly, the trader is inclined to post further from the mid to incur less impact. On average, this leads to larger profits per fill, but fewer fills, which explains the more aggressive posting towards the end of the time window. 

Similar to Table \ref{tab:eta_summary}, we also show summary statistics of $1000$ Monte Carlo simulations in Table \ref{tab:m_summary} as $m$ is varied. Since $k=48$, the table shows results for $m<k$, $m\approx k$, and $m>k$. When the sensitivity $m$ is higher, the trader is incentivised to post deeper in the book to accrue less impact. Therefore, more spread is made per trade and terminal inventory is valued at a better price due to the lower impact, and this increases P\&L. However, this comes at a cost of reduced fills, leading final inventory and trading time to both be higher when $m>k$ than they would otherwise be. 

\begin{table}[H]
\begin{center}
\caption{Summary statistics by $m$.}\label{tab:m_summary}
\vskip-0.2cm
\small
\setlength{\tabcolsep}{6pt}
\renewcommand{\arraystretch}{1.15}
\setlength{\dashlinedash}{0.4pt}
\setlength{\dashlinegap}{2pt}
\arrayrulecolor{gray!60}
\makebox[\textwidth][c]{%
\begin{tabular}{l:c:c:c:c}
\toprule
\textbf{$m$}
& \begin{tabular}{@{}c@{}}\textbf{Final inventory}\\ \textbf{(thousand shares)}\end{tabular}
& \begin{tabular}{@{}c@{}}\textbf{P\&L}\\ \textbf{(thousand \$)}\end{tabular}
& \begin{tabular}{@{}c@{}}\textbf{Implementation shortfall}\\ \textbf{(thousand \$)}\end{tabular}
& \begin{tabular}{@{}c@{}}\textbf{Trading time}\\ \textbf{(s)}\end{tabular} \\
\midrule
\arrayrulecolor{gray!60}

$\mathbf{40}$
& \begin{tabular}{@{}c@{}}$0.044$\\$(0.031,\ 0.057)$\end{tabular}
& \begin{tabular}{@{}c@{}}$0.469$\\$(0.442,\ 0.496)$\end{tabular}
& \begin{tabular}{@{}c@{}}$-0.024$\\$(-0.025,\ -0.022)$\end{tabular}
& \begin{tabular}{@{}c@{}}$266.987$\\$(265.213,\ 268.762)$\end{tabular}
\\
\hdashline

$\mathbf{50}$
& \begin{tabular}{@{}c@{}}$0.071$\\$(0.054,\ 0.088)$\end{tabular}
& \begin{tabular}{@{}c@{}}$0.766$\\$(0.737,\ 0.796)$\end{tabular}
& \begin{tabular}{@{}c@{}}$-0.039$\\$(-0.040,\ -0.037)$\end{tabular}
& \begin{tabular}{@{}c@{}}$275.043$\\$(273.607,\ 276.480)$\end{tabular}
\\
\hdashline

$\mathbf{60}$
& \begin{tabular}{@{}c@{}}$0.081$\\$(0.063,\ 0.099)$\end{tabular}
& \begin{tabular}{@{}c@{}}$0.939$\\$(0.908,\ 0.969)$\end{tabular}
& \begin{tabular}{@{}c@{}}$-0.047$\\$(-0.049,\ -0.046)$\end{tabular}
& \begin{tabular}{@{}c@{}}$276.956$\\$(275.618,\ 278.295)$\end{tabular}
\\

\arrayrulecolor{black}
\bottomrule
\end{tabular}%
}
\end{center}
\end{table}

\subsection{Transient passive impact} \label{sec:transient}

Our main model in Section \ref{sec:model} treats passive impact as permanent: passive pressure generated at time \(t\) affects the midprice level thereafter. A natural extension is to let this impact decay over time as in \citet{Obizhaeva2013Optimal}. To model this, we introduce a transient impact state \((I_t)_{t\in [0,T]}\) which is given by
\begin{equation}
    I_t
    =
    -\rho \int_{0}^{t}I_u\,du
    +
    \eta  \int_{0}^{t}e^{-k\delta_u}\mathbf 1_{\{Q_{u-}>0\}}\,du.
\end{equation}
Here \(\rho\) is the resilience rate of passive impact. Larger \(\rho\) means that passive pressure decays more quickly. Then, instead of \eqref{eq:S}, the observed midprice satisfies
\[
    S_t
    = S_0 + 
    \sigma\,W_t
    +\rho\int_{0}^{t} I_u\,du
    -\eta  \int_{0}^{t}e^{-k\delta_u}\mathbf 1_{\{Q_{u-}>0\}}\,du.
\]
Thus, the value function is of the form $u(t,x,q,s,i)$, defined by
\begin{equation}\label{eq:optimisation-problem-transient}
     u(t,x,q,s, i)
    :=
    \sup_{\delta \in \mathcal A}
    \mathbb E_{t,x,q,s,i}
    \bigg[
        X_T + Q_T S_T 
        - \phi \int_t^T Q_u^2\,du
        - \alpha Q_T^{2}
    \bigg],
\end{equation}
and the HJB equation \eqref{eq:hjb} becomes
\begin{equation}
\label{eq:hjb-transient}
\begin{aligned}
    0
    =
    &\frac{\partial u}{\partial t}(t, x, q, s, i)
    +\frac12\sigma^2\frac{\partial^{2} u}{\partial s^{2}}(t, x, q, s, i)
    +
    \rho i\,\frac{\partial u}{\partial s}(t, x, q, s, i)
    -
    \rho i\,\frac{\partial u}{\partial i}(t, x, q, s, i)
    -
    \phi q^2\\
    &+
    \mathbf 1_{\{q>0\}}
    \sup_{\delta}
    \bigg\{
        \lambda e^{-k\delta}
        \big(
             u(t,x+s+\delta,q-1,s,i)
            -
             u(t,x,q,s,i)
        \big) \\
    &\qquad\qquad\qquad\ 
        -
         \eta e^{-k\delta}
        \frac{\partial u}{\partial s}(t, x, q, s, i)
        +
        \eta e^{-k\delta}
        \frac{\partial u}{\partial i}(t, x, q, s, i)
    \bigg\}.
\end{aligned}
\end{equation}

\begin{theorem}\label{thm:transient}
Let $m=k$. For \(q=1,\dots,q_0\), the optimal quote is
\begin{equation}
\label{eq:optimal-quote-transient}
    \delta^\star(t,q,i)
    =
    \frac1k
    +
    \psi(t,q,i)-\psi(t,q-1,i)
    +
    \frac{\eta}{\lambda} q
    -
    \frac{\eta}{\lambda}\frac{\partial\psi}{\partial i}(t,q,i),
\end{equation}
where \(\psi(t,q,i)\) solves the nonlinear system
\begin{equation}
\label{eq:hjb-transient-optimised}
\begin{aligned}
    0
    =
    &\frac{\partial\psi}{\partial t}(t,q,i)
    +
    \rho i\bigg(q-\frac{\partial\psi}{\partial i}(t,q,i)\bigg)
    -
    \phi q^2\\
    &+
    \mathbf 1_{\{q>0\}}
    \frac{e^{-1}\lambda}{k}
    \exp\bigg(
        -k\bigg(
            \psi(t,q,i)-\psi(t,q-1,i)
            +
            \frac{\eta}{\lambda} q
            -
            \frac{\eta}{\lambda}\frac{\partial\psi}{\partial i}(t,q,i)
        \bigg)
    \bigg)
\end{aligned}
\end{equation}
with boundary and terminal conditions
\begin{equation}\label{eq:boundary-terminal-transient}
    \psi(t,0,i)=0,
    \qquad
    \psi(T,q,i)=-\alpha q^2.
\end{equation}
\end{theorem}

The system \eqref{eq:hjb-transient-optimised} is substantially less tractable than in the permanent-impact case. The exponential transform used before no longer yields a triangular linear system of ODEs. Instead, one obtains a nonlinear first-order PDE in the additional state variable \(I\).

\begin{figure}[htb]
    \centering
    \includegraphics[width=0.7\textwidth]{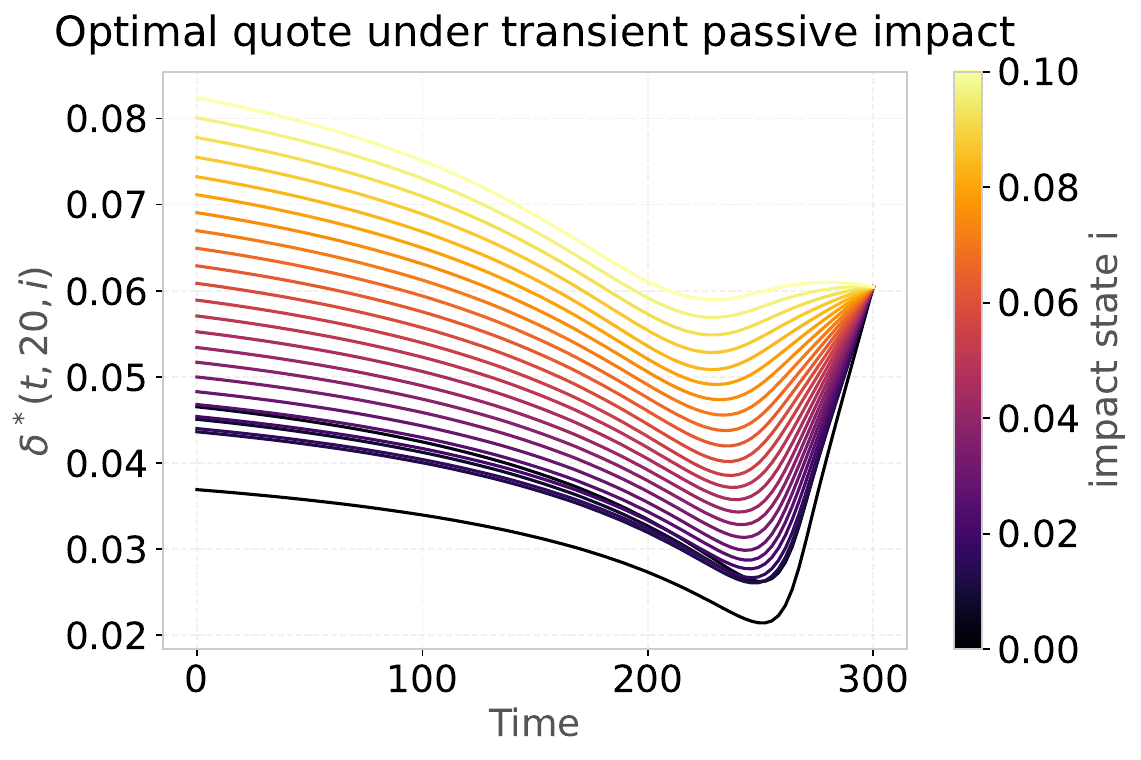}
    \caption{Optimal quote depth $\delta^{\star}(t, 20,i)$ over time for different impact states $i$ with $\rho=0.01$.}
    \label{fig:quotes_transient}
\end{figure}

Figure \ref{fig:quotes_transient} compares the optimal quote depth $\delta^{\star}(t, 20, i)$ across different values of current midprice dislocation due to price impact $i$, with lighter colours indicating a more impacted midprice. The inventory is fixed at $q=20$. We see that when there is more impact, the trader optimally quotes further from the midprice to allow the impact to dissipate, rather than continuing to trade at an adverse price. As the end of the trading period approaches, this strategy of waiting becomes less feasible and hence the difference in quotes due to $i$ narrows.

Similar to Table \ref{tab:eta_summary}, we also show summary statistics of $1000$ Monte Carlo simulations in Table \ref{tab:rho_summary} as $\rho$ is varied and the other parameters equal to the values given in Table \ref{tab:params}. Also similar to the effect of increasing $\eta$ in Table \ref{tab:eta_summary}, when $\rho$ is lower and thus impact decays more slowly, the trader achieves worse P\&L. However, because faster impact decay allows the price to recover between trades, unsold inventory is marked at a more favourable terminal price which reduces the trader's urgency. This results in higher final inventory and slower trading times on average. Note that the $\rho=0$ case corresponds to permanent price impact.

\begin{table}[H]
\begin{center}
\caption{Summary statistics by $\rho$.}\label{tab:rho_summary}
\vskip-0.2cm
\small
\setlength{\tabcolsep}{6pt}
\renewcommand{\arraystretch}{1.15}
\setlength{\dashlinedash}{0.4pt}
\setlength{\dashlinegap}{2pt}
\arrayrulecolor{gray!60}
\makebox[\textwidth][c]{%
\begin{tabular}{l:c:c:c:c}
\toprule
\textbf{$\rho$}
& \begin{tabular}{@{}c@{}}\textbf{Final inventory}\\ \textbf{(thousand shares)}\end{tabular}
& \begin{tabular}{@{}c@{}}\textbf{P\&L}\\ \textbf{(thousand \$)}\end{tabular}
& \begin{tabular}{@{}c@{}}\textbf{Implementation shortfall}\\ \textbf{(thousand \$)}\end{tabular}
& \begin{tabular}{@{}c@{}}\textbf{Trading time}\\ \textbf{(s)}\end{tabular} \\
\midrule
\arrayrulecolor{gray!60}

$\mathbf{0.0}$
& \begin{tabular}{@{}c@{}}$0.059$\\$(0.044,\ 0.074)$\end{tabular}
& \begin{tabular}{@{}c@{}}$0.717$\\$(0.688,\ 0.747)$\end{tabular}
& \begin{tabular}{@{}c@{}}$-0.036$\\$(-0.038,\ -0.035)$\end{tabular}
& \begin{tabular}{@{}c@{}}$273.812$\\$(272.312,\ 275.312)$\end{tabular}
\\
\hdashline

$\mathbf{0.01}$
& \begin{tabular}{@{}c@{}}$0.104$\\$(0.084,\ 0.124)$\end{tabular}
& \begin{tabular}{@{}c@{}}$0.981$\\$(0.949,\ 1.013)$\end{tabular}
& \begin{tabular}{@{}c@{}}$-0.049$\\$(-0.051,\ -0.048)$\end{tabular}
& \begin{tabular}{@{}c@{}}$282.984$\\$(281.884,\ 284.083)$\end{tabular}
\\
\hdashline

$\mathbf{0.02}$
& \begin{tabular}{@{}c@{}}$0.094$\\$(0.075,\ 0.113)$\end{tabular}
& \begin{tabular}{@{}c@{}}$1.047$\\$(1.016,\ 1.079)$\end{tabular}
& \begin{tabular}{@{}c@{}}$-0.053$\\$(-0.054,\ -0.051)$\end{tabular}
& \begin{tabular}{@{}c@{}}$280.880$\\$(279.676,\ 282.084)$\end{tabular}
\\

\arrayrulecolor{black}
\bottomrule
\end{tabular}%
}
\end{center}
\end{table}

\subsection{Target execution schedules}\label{sec:target}

In the baseline model, the trader's inventory risk is given by the running penalty
\[
    \phi\int_t^T Q_u^2\,du.
\]
This corresponds to penalising deviations from the zero-inventory target. A natural extension is to replace this by a deterministic target trajectory (see \cite{Cartea2015Algorithmic}, Section 8.5). Let \(f:[0,T]\to\mathbb R\) be a bounded deterministic function and consider the objective
\begin{equation}\label{eq:target-objective}
     u_f(t,x,q,s)
    :=
    \sup_{\delta \in \mathcal A}
    \mathbb E_{t,x,q,s}
    \bigg[
        X_T + Q_T S_T
        - \phi \int_t^T \big(Q_u-f(u)\big)^2\,du
        - \alpha Q_T^{2}
    \bigg].
\end{equation}
The case \(f\equiv0\) recovers the baseline objective \eqref{eq:optimisation-problem}. A TWAP target can be represented by the decreasing deterministic trajectory 
\begin{equation} \label{eq:TWAP}
    f(t)=q_0\bigg(1-\frac{t}{T}\bigg),
\end{equation}
penalising deviations from a linear liquidation schedule. The HJB equation \eqref{eq:hjb} then becomes
\begin{equation}
\label{eq:hjb-target-full}
\begin{aligned}
    0
    =
    &\frac{\partial u_f}{\partial t}(t,x,q,s)
    + \frac{1}{2}\sigma^2 \frac{\partial^{2} u_f}{\partial s^{2}}(t,x,q,s)
    - \phi \big(q-f(t)\big)^2 \\
    &+ \mathbf 1_{\{q>0\}}
    \sup_{\delta}
    \bigg\{
        \lambda e^{-k\delta}
        \bigg(
             u_f(t,x+s+\delta,q-1,s) - u_f(t,x,q,s)
            - \frac{\eta}{\lambda} \frac{\partial u_f}{\partial s}(t,x,q,s)
        \bigg)
    \bigg\},
\end{aligned}
\end{equation}
with terminal condition
\[
    u_f(T,x,q,s)=x+qs-\alpha q^2.
\]
Notice that, if \(f\) is not identically zero, then \(u_f(t,x,0,s)\neq x\). Even after full liquidation, the trader may still incur a penalty for deviating from the prescribed target trajectory.

\begin{theorem}\label{thm:target}
Let \(m=k\) and \(f:[0,T]\to\mathbb R\) be bounded and deterministic. For \(q=0,\ldots,q_0\), define
\[
    a_q(t):=k\phi\big(q-f(t)\big)^2,
    \qquad
    b_q:=e^{-1}\lambda e^{-k\frac{\eta}{\lambda}q}.
\]
Let \(\omega_f(t,q)\) be defined recursively by
\begin{equation}
\label{eq:omega-target-recursive}
\begin{aligned}
    \omega_f(t,q)
    =
    \exp\bigg(
        -\int_t^T a_q(r)\,dr
    \bigg)
    e^{-k\alpha q^2}
    +
    \mathbf 1_{\{q>0\}}
    b_q
    \int_t^T
        \exp\bigg(
            -\int_t^r a_q(v)\,dv
        \bigg)
        \omega_f(r,q-1)\,dr .
\end{aligned}
\end{equation}
Then, for \(q=1,\dots,q_0\), the optimal quote is
\begin{equation}
\label{eq:optimal-quote-target}
    \delta_f^\star(t,q)
    =
    \frac1k
    +
    \frac1k
    \log\bigg(
        \frac{\omega_f(t,q)}{\omega_f(t,q-1)}
    \bigg)
    +
    \frac{\eta}{\lambda}q .
\end{equation}
\end{theorem}
The target trajectory \(f\) therefore changes the value function through the time-dependent diagonal terms \(a_q(t)=k\phi(q-f(t))^2\), but it does not change the optimal quote formula except through the ratio \(\omega_f(t,q)/\omega_f(t,q-1)\). For a general deterministic target \(f\), the constant matrix exponential in Theorem~\ref{thm:main} is replaced by the recursive representation \eqref{eq:omega-target-recursive}. If \(f\equiv0\), then \(a_q(t)=k\phi q^2\) is constant, \(\omega_f(t,0)=1\), and \eqref{eq:omega-target-recursive} reduces to the baseline recursion \eqref{eq:omega-recursive-solution}.

\begin{figure}[htb]
    \centering
    \includegraphics[width=0.75\textwidth]{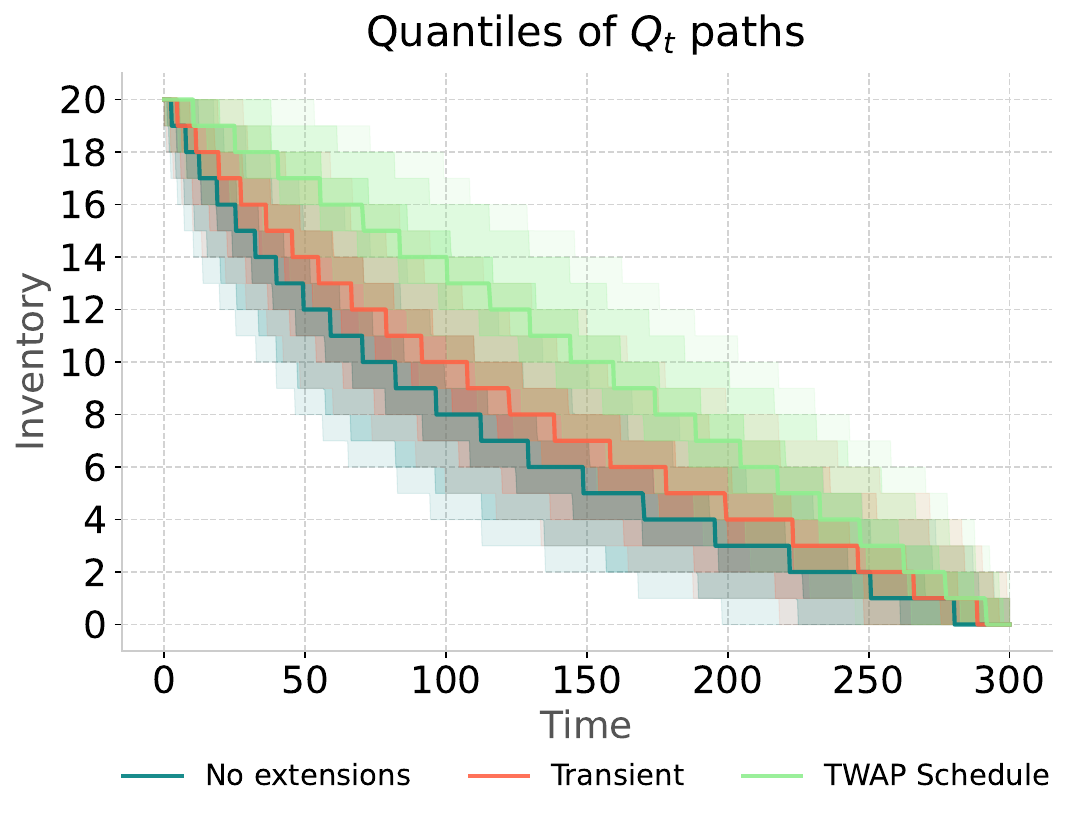}
    \caption{The median path of the inventory $Q_{t}$ in the baseline model with permanent impact (blue), the transient impact extension (red), and the TWAP-target execution schedule model with permanent impact (green) with 25-75\%, 5-95\%, and 1-99\% quantiles shaded in progressively lighter shades.}
    \label{fig:inventory_quantilemap}
\end{figure}

Figure \ref{fig:inventory_quantilemap} shows the median path of inventory $Q$, along with a heatmap of trajectories for the baseline model with no extensions, the model with transient impact, and the model with a TWAP target execution schedule. Indeed, we note that in all of our models, the execution schedule is stochastic due to the random fills, in contrast to classical optimal execution literature based on market orders. 

Convexity of the execution schedule is a key result of \citet{Almgren2001Optimal}. We see that when price impact is transient, the optimal execution schedule is closer to a TWAP, which would correspond to a straight line, than in the case of permanent impact. This is because transient impact creates an incentive to space out orders and allow impact to decay. By contrast, under permanent impact, the trader prefers to reduce inventory risk by executing more quickly, since the total impact cost for a given number of executed orders is unchanged by the timing of those orders. The target execution schedule associated with the TWAP target \eqref{eq:TWAP} is therefore much closer to a linear schedule, as expected.

We can also observe in Figure \ref{fig:inventory_quantilemap} that, for the baseline model with no extensions, the median trajectory is qualitatively close to classical deterministic liquidation profiles, which consider only market orders. However, the reasoning in both cases is different. In the market-order execution case, convexity is mainly driven by the running inventory penalty term (see \cite{Almgren2001Optimal}), while a large temporary impact coefficient reduces convexity, and permanent impact factors out of the objective functional. In the limit-order case, convexity is driven by the running inventory term and permanent passive impact.

\section{Concluding Remarks} \label{sec-rem}

This paper develops a mesoscopic framework for optimal execution with passive orders. Its main objective is to bridge two levels of description that are typically treated separately. At the microstructural level, limit-order submissions, cancellations, and executions contribute to order-flow imbalance and may consequently affect prices. At the execution level, however, passive trading is commonly represented solely through a fill probability and the spread earned conditional on execution \cite{Avellaneda2008High, Cartea2015Algorithmic, Gueant2012Optimal}. The present model introduces an intermediate description in which the detailed sequence of order submissions, cancellations, and repostings is replaced by an average passive impact rate controlled by the quote distance. This formulation preserves a clear microstructural interpretation while remaining analytically tractable.

The main conceptual implication is that passive execution can move mesoscopic prices, which in turn affects the optimal execution strategy. A limit order may provide liquidity locally, but the tactical process through which a larger order is executed can generate persistent directional pressure. Incorporating this pressure changes the optimal balance between spread capture, execution probability and inventory risk. In particular, permanent passive impact affects the optimal quoting policy and the shape of the liquidation trajectory. 

After the initial preparation of this manuscript, a revised version of the contemporaneous work of \citet{ouazzani2024theory} appeared. Their order-book framework shows how liquidity-state-dependent information in market-order flow can generate non-trivial impact curves for passive metaorders. This result reinforces the premise of the present paper that passive execution is not impact-free. The two approaches are complementary: their theory provides a microscopic, queue-based foundation, whereas our work develops a mesoscopic optimal-execution formulation in which quote distance controls both fill intensity and passive-impact accumulation. This level of description is particularly natural for fragmented or quote-driven markets such as FX.

There are several directions for further work. First, the empirical results support the ingredients of the reduced-form model rather than providing a causal estimate of the impact generated by an individual trader's passive strategy. The MLOFI coefficients measure contemporaneous price responses to aggregate order-book events, while the FX fill intensities are inferred from public L2 quote and trade data using hypothetical passive quotes. A stronger empirical test would require trader-tagged or order-level data linking submissions, cancellations, fills, and subsequent price responses.

Second, the exponential specifications are intentionally parsimonious. Fill and impact parameters are likely to depend on the prevailing spread, queue position, displayed depth, volatility, imbalance, order size, intraday seasonality, and recent order flow. Allowing $\lambda$, $k$, $\eta$, and $m$ to depend on these state variables would yield a more realistic control problem but would also remove much of the current analytical structure. The transient formulation likewise raises an empirical question concerning the appropriate choice of decay kernel.

Finally, a practical execution algorithm rarely relies exclusively on one-unit passive orders. Natural extensions include the joint choice of quote distance and displayed size, simultaneous passive and aggressive execution, multiple venues with heterogeneous fill and impact characteristics, and the incorporation of explicit adverse-selection signals. At the portfolio level, passive self-impact should interact with cross-impact across related instruments. These directions would move the framework closer to a full execution architecture in which passive and aggressive actions are selected jointly. The present model, however, retains the advantage of significant tractability, allowing for the derivation of optimal execution policies.

\section{Proof of Theorem \ref{thm:main}}\label{sec:proof-main}

This section is devoted to the proof of the main theorem.
\begin{proof}[Proof of Theorem~\ref{thm:main}]
Recall the HJB equation \eqref{eq:hjb}. We use the ansatz
\begin{equation}\label{eq:ansatz}
     u(t,x,q,s) = x + qs + \theta(t,q).
\end{equation}
Then \eqref{eq:hjb} reduces to
\begin{equation}
\label{eq:hjb-theta-exponential}
    \frac{\partial \theta}{\partial t}(t,q)
    - \phi q^2
    +
    \mathbf 1_{\{q>0\}}
    \sup_{\delta}
    \lambda e^{-k\delta}
    \bigg(
        \delta + \theta(t,q-1) - \theta(t,q) - \frac{\eta}{\lambda} q
    \bigg)
    = 0,
\end{equation}
with boundary and terminal conditions
\begin{equation}\label{eq:boundary-terminal-theta}
    \theta(t,0)=0,
    \qquad
    \theta(T,q) = -\alpha q^{2}.
\end{equation}
The one-step optimisation problem is therefore
\begin{equation}
\label{eq:hamiltonian-exponential-def}
    H\bigg(\theta(t,q)-\theta(t,q-1)+\frac{\eta}{\lambda} q\bigg)
    :=
    \sup_{\delta}
    \bigg\{
        \lambda e^{-k\delta}
        \bigg(
            \delta-\theta(t,q)+\theta(t,q-1)-\frac{\eta}{\lambda} q
        \bigg)
    \bigg\}.
\end{equation}
The first-order condition associated with \eqref{eq:hamiltonian-exponential-def} yields
\begin{equation}
\label{eq:delta-star-theta}
    \delta^\star(t,q)
    =
    \frac{1}{k}
    +
    \theta(t,q)-\theta(t,q-1)+\frac{\eta}{\lambda} q.
\end{equation}
Substituting \eqref{eq:delta-star-theta} into \eqref{eq:hjb-theta-exponential}, we obtain
\begin{equation}
\label{eq:hjb-theta-final}
    \frac{\partial \theta}{\partial t}(t,q)
    - \phi q^2
    +
    \mathbf 1_{\{q>0\}}
    \frac{e^{-1}\lambda}{k}
    \exp\bigg(
        -k\bigg(
            \theta(t,q)-\theta(t,q-1)+\frac{\eta}{\lambda} q
        \bigg)
    \bigg)
    = 0.
\end{equation}
This nonlinear equation is linearised by the transform
\be\label{eq:omega}
    \omega(t,q):=e^{k\theta(t,q)}.
\ee
Indeed, combining \eqref{eq:boundary-terminal-theta}, \eqref{eq:hjb-theta-final}, and \eqref{eq:omega} gives the triangular linear system
\begin{equation}
\label{eq:omega-linear-system}
    \frac{\partial \omega}{\partial t}(t,q)
    =
    k \phi q^2 \omega(t,q)
    -
    e^{-1}\lambda e^{-k\frac{\eta}{\lambda} q}\omega(t,q-1),
\end{equation}
with boundary and terminal conditions
\begin{equation}\label{eq:boundary-terminal-omega}
    \omega(t,0)=1,
    \qquad
    \omega(T,q)=e^{-k\alpha q^2}.
\end{equation}
Solving \eqref{eq:omega-linear-system}--\eqref{eq:boundary-terminal-omega} backward from \(T\)
gives the matrix representation \eqref{eq:omega-matrix-representation}--\eqref{eq:omega-matrix-A}. Finally, using
\eqref{eq:delta-star-theta} together with \eqref{eq:omega}, we obtain
\begin{equation}
    \delta^\star(t,q)
    =
    \frac{1}{k}
    +
    \frac{1}{k}
    \log\bigg(
        \frac{\omega(t,q)}{\omega(t,q-1)}
    \bigg)
    +
    \frac{\eta}{\lambda} q,
\end{equation}
which is exactly \eqref{eq:optimal-quote-omega}. This proves the claim.
\end{proof}

\section{Proofs of Theorems \ref{thm:decay}, \ref{thm:transient}, and \ref{thm:target}} \label{sec:proofs-extensions}

This section focuses on the proofs of our results on the extension of the model.
\begin{proof}[Proof of Theorem~\ref{thm:decay}]
As in the proof of Theorem~\ref{thm:main}, we use the ansatz
\[
     u(t,x,q,s)=x+qs+\theta(t,q).
\]
Then, writing
\be\label{eq:Deltaq}
    \Delta_q(t):=\theta(t,q)-\theta(t,q-1),
\ee
and 
\be\label{eq:Hq}
    H_q(\Delta_q)
    :=
        \lambda e^{-k\delta}(\delta-\Delta_q)
        -
        \eta q e^{-m\delta},
\ee
the HJB equation \eqref{eq:hjb-diffdecay} reduces to
\be\label{eq:hjb-diffdecay-reduced}
    \frac{\partial\theta}{\partial t}(t,q)-\phi q^2
    +
    \mathbf 1_{\{q>0\}}
    \sup_{\delta}
    H_q(\Delta_q)
    =
    0.
\ee
with boundary and terminal conditions as in \eqref{eq:boundary-terminal-diffdecay}. The corresponding first-order condition for a maximiser of \eqref{eq:Hq} is given by
\[
    \lambda e^{-k\delta}
    \big(
        1-k(\delta-\Delta_q)
    \big)
    +
    \eta q m e^{-m\delta}
    =
    0
\]
which, by dividing by $\lambda e^{-k\delta}$, is equivalent to
\be\label{eq:foc}
    k(\delta-\Delta_q)-1
    =
    \frac{\eta}{\lambda} q m e^{-(m-k)\delta}.
\ee
To solve \eqref{eq:foc} for $\delta$, set
\[
    y:=k(\delta-\Delta_q)-1
\]
which is equivalent to
\be\label{eq:deltay}
    \delta
    =
    \Delta_q+\frac{1+y}{k}.
\ee
Substituting \eqref{eq:deltay} into \eqref{eq:foc} gives
\be\label{eq:yeq}
    y
    =
    \frac{\eta}{\lambda} q m
    \exp\bigg(
        -(m-k)\bigg(\Delta_q+\frac{1+y}{k}\bigg)
    \bigg).
\ee
Multiplying \eqref{eq:yeq} first by \(\exp\big(\frac{m-k}{k}y\big)\) and then by \((m-k)/k\), we obtain
\be\label{eq:lambertW-eq}
    \frac{m-k}{k}y
    \exp\bigg(\frac{m-k}{k}y\bigg)
    =
    \frac{m-k}{k}\frac{\eta}{\lambda} q m
    \exp\bigg(
        -\frac{m-k}{k}(1+k\Delta_q)
    \bigg).
\ee
Recall that \(W_0\) denotes the principal real branch of the Lambert \(W\) function, defined by
\begin{equation}
    W_0(z)e^{W_0(z)}=z,
    \qquad z\in[-e^{-1},\infty)
\end{equation}
with \(W_0(z)\geq -1\). On the interval \(z\in[-e^{-1},0)\), the second real branch is denoted by \(W_{-1}\), and satisfies \(W_{-1}(z)\leq -1\).

\textbf{Case 1.} Suppose that \(m>k\). Then \eqref{eq:lambertW-eq} is positive and can be solved in terms of the principal branch \(W_0\) of the Lambert \(W\) function from \eqref{eq:LambertW}. Namely,
\be\label{eq:solution-y}
    y
    =
    \frac{k}{m-k}
    W_0\bigg(
        \frac{m-k}{k}\frac{\eta}{\lambda} q m
        \exp\bigg(
            -\frac{m-k}{k}(1+k\Delta_q)
        \bigg)
    \bigg).
\ee
By \eqref{eq:deltay} and \eqref{eq:solution-y}, the optimiser is therefore
\begin{equation}
\label{eq:delta-star-different-exponents}
    \delta_q^\star(\Delta_q)
    =
    \Delta_q+\frac1k
    +
    \frac{1}{m-k}
    W_0\bigg(
        \frac{m-k}{k}\frac{\eta}{\lambda} q m
        \exp\bigg(
            -\frac{m-k}{k}(1+k\Delta_q)
        \bigg)
    \bigg)
\end{equation}
which, by \eqref{eq:Phi-q-def} and \eqref{eq:Deltaq}, is exactly \eqref{eq:optimal-quote-diffdecay}. Moreover, evaluating \eqref{eq:Hq} at the optimiser gives
\be
\begin{aligned}\label{eq:Hq-opt}
H_q(\delta_q^\star,\Delta_q)
&=
\lambda e^{-k\delta_q^\star}(\delta_q^\star-\Delta_q)
-
\frac{\eta}{\lambda} q\lambda e^{-m\delta_q^\star}.
\end{aligned}
\ee
Rearranging the first-order condition \eqref{eq:foc} into
\be\label{eq:foc-rearranged}
    \frac{\eta}{\lambda} q e^{-m\delta_q^\star}
    =
    \frac{1}{m}e^{-k\delta_q^\star}
    \big(
        k(\delta_q^\star-\Delta_q)-1
    \big),
\ee
we obtain from \eqref{eq:Hq-opt} and \eqref{eq:foc-rearranged}
\be
\begin{aligned}
H_q(\delta_q^\star,\Delta_q)
&=
\lambda e^{-k\delta_q^\star}
\bigg(
    \delta_q^\star-\Delta_q
    -
    \frac{1}{m}
    \bigg(
        k(\delta_q^\star-\Delta_q)-1
    \bigg)
\bigg).
\end{aligned}
\ee
Using \eqref{eq:Phi-q-def}, \eqref{eq:Deltaq}, and \eqref{eq:delta-star-different-exponents}, this simplifies to
\be\label{eq:Hq-opt-simplified}
H_q(\delta_q^\star,\Delta_q)
=
\lambda e^{-k\delta_q^\star}
\bigg(
    \frac1k+\frac{1}{m}\Phi_q(\Delta_q)
\bigg).
\ee
Finally, using \eqref{eq:Phi-q-def}, \eqref{eq:Deltaq}, and \eqref{eq:delta-star-different-exponents} again,
\be\label{eq:aux}
e^{-k\delta_q^\star}
=
\exp\bigg(
    -1
    -k\Delta_q
    -
    \frac{k}{m-k}\Phi_q(\Delta_q)
\bigg).
\ee
Thus, by \eqref{eq:Hq-opt-simplified} and \eqref{eq:aux},
\[
\sup_\delta H_q(\Delta_q)
=
\lambda
\bigg(
    \frac1k+\frac{1}{m}\Phi_q(\Delta_q)
\bigg)
\exp\bigg(
    -1
    -k\Delta_q
    -
    \frac{k}{m-k}\Phi_q(\Delta_q)
\bigg),
\]
and substituting this into \eqref{eq:hjb-diffdecay-reduced} gives \eqref{eq:theta-diffdecay}. This proves the claim.

\textbf{Case 2.} Suppose now that \(m<k\). Then the argument of the Lambert
\(W\) function in \eqref{eq:lambertW-eq} is negative. For \(m\) sufficiently
close to \(k\), this argument belongs to \((-e^{-1},0)\), and hence both real
branches \(W_0\) and \(W_{-1}\) are available. For $m$ sufficiently close to zero the argument also belongs to \((-e^{-1},0)\). Thus \eqref{eq:lambertW-eq} has two real solutions, given by
\be\label{eq:lambertW-solutions}
    y_\ell
    =
    \frac{k}{m-k}
    W_\ell\bigg(
        \frac{m-k}{k}\frac{\eta}{\lambda} q m
        \exp\bigg(
            -\frac{m-k}{k}(1+k\Delta_q)
        \bigg)
    \bigg),
    \qquad \ell\in\{0,-1\}.
\ee
Correspondingly, by \eqref{eq:deltay}, the two critical points are
\be\label{eq:critical-points}
    \delta_\ell
    =
    \Delta_q+\frac1k
    +
    \frac{1}{m-k}
    W_\ell\bigg(
        \frac{m-k}{k}\frac{\eta}{\lambda} q m
        \exp\bigg(
            -\frac{m-k}{k}(1+k\Delta_q)
        \bigg)
    \bigg),
    \qquad \ell\in\{0,-1\}.
\ee
We now determine which branch gives the maximiser. Differentiating
\eqref{eq:Hq} twice gives
\be\label{eq:Hq''}
    \frac{\partial^2 H_q}{\partial \delta^2}(\delta)
    =
    \lambda e^{-k\delta}
    \big(
        k^2(\delta-\Delta_q)-2k
    \big)
    -
    \eta q m^2 e^{-m\delta}.
\ee
At a critical point, the first-order condition \eqref{eq:foc} gives
\be\label{eq:foc-rearranged2}
    \frac{\eta}{\lambda} q m e^{-m\delta}
    =
    e^{-k\delta}
    \big(
        k(\delta-\Delta_q)-1
    \big).
\ee
Hence, at a critical point, by \eqref{eq:Hq''} and \eqref{eq:foc-rearranged2},
\[
\begin{aligned}
    \frac{\partial^2 H_q}{\partial \delta^2}(\delta)
    &=
    \lambda e^{-k\delta}
    \Big(
        k^2(\delta-\Delta_q)-2k
        -
        m\big(k(\delta-\Delta_q)-1\big)
    \Big) \\
    &=
    \lambda e^{-k\delta}
    \Big(
        (k-m)\big(k(\delta-\Delta_q)-1\big)-m
    \Big).
\end{aligned}
\]
Using \(y=k(\delta-\Delta_q)-1\), this becomes
\be\label{eq:Hq''-simplified}
    \frac{\partial^2 H_q}{\partial \delta^2}(\delta)
    =
    \lambda e^{-k\delta}
    \big(
        (k-m)y-m
    \big).
\ee
Plugging \eqref{eq:lambertW-solutions} into \eqref{eq:Hq''-simplified}, we obtain
\be\label{eq:Hq''-lambert}
    \frac{\partial^2 H_q}{\partial \delta^2}(\delta_\ell)
    =
    -\lambda e^{-k\delta_\ell}
    \big(
        k W_\ell(\cdots)+m
    \big),
    \qquad \ell\in\{0,-1\}.
\ee
For \(\ell=0\), we have \(W_0(\cdots)\in(-1,0)\). Moreover, as
\(m\uparrow k\), the argument of \(W_0\) tends to zero and therefore
\(W_0(\cdots)\to0\). Hence, by \eqref{eq:Hq''-lambert}, for \(m<k\) sufficiently close to \(k\), 
\[
    \frac{\partial^2 H_q}{\partial \delta^2}(\delta_0)<0,
\]
so the \(W_0\)-branch gives a local maximum. For \(\ell=-1\), we have \(W_{-1}(\cdots)\leq -1\). Since \(m<k\), by \eqref{eq:Hq''-lambert},
\[
    \frac{\partial^2 H_q}{\partial \delta^2}(\delta_{-1})>0,
\]
so the \(W_{-1}\)-branch gives a local minimum.

It remains to compare the local maximum with the boundary as
\(\delta\to\infty\). Since \(m<k\), the second exponential decays more slowly
than the first one, and therefore
\[
    H_q(\delta;\Delta_q)
    =
    \lambda e^{-k\delta}(\delta-\Delta_q)
    -
    \eta q e^{-m\delta}
    \longrightarrow 0,
    \qquad
    \text{as } \delta\to\infty.
\]
More precisely, the convergence is from below for large \(\delta\). On the
other hand, the value of the Hamiltonian at the \(W_0\)-critical point depends
continuously on \(m\), and as \(m\uparrow k\) it converges to the maximised
Hamiltonian in the equal-decay case \(m=k\), which is strictly positive.
Consequently, for \(m<k\) sufficiently close to \(k\),
\[
    H_q(\delta_0;\Delta_q)>0.
\]
Similarly, by \eqref{eq:Hq},
as \(\delta\to-\infty\), both terms tend to \(-\infty\), and hence
\[
    H_q(\delta;\Delta_q)\to-\infty,\qquad\text{as } \delta\to-\infty.
\]
Thus the \(W_0\)-critical point is a local maximum that dominates the left and right
boundary values, while the \(W_{-1}\)-critical point is a local minimum. Hence, the global maximiser is given in terms of the principal branch \(W_0\), and by \eqref{eq:critical-points} the optimal quote is given by \eqref{eq:optimal-quote-diffdecay}.

\end{proof}

\begin{proof}[Proof of Theorem~\ref{thm:transient}]
Similar to the proof of Theorem \ref{thm:main}, we use the ansatz
\[
    u(t,x,q,s,i)=x+qs+\psi(t,q,i).
\]
This reduces the HJB equation \eqref{eq:hjb-transient} to 
\begin{equation}
\label{eq:hjb-transient-psi}
\begin{aligned}
    0
    =
    &\frac{\partial\psi}{\partial t}(t,q,I)
    +
    \rho I\bigg(q-\frac{\partial\psi}{\partial i}(t,q,I)\bigg)
    -
    \phi q^2\\
    &+
    \mathbf 1_{\{q>0\}}
    \sup_{\delta}\bigg\{
    \lambda e^{-k\delta}
    \bigg(
        \delta
        -
        \bigg(
            \psi(t,q,I)-\psi(t,q-1,I) + \frac{\eta}{\lambda}q
            -
            \frac{\eta}{\lambda}\frac{\partial\psi}{\partial i}(t,q,I)
        \bigg)
    \bigg)\bigg\}.
\end{aligned}
\end{equation}
The maximiser is given by \eqref{eq:optimal-quote-transient}, which substituted into \eqref{eq:hjb-transient-psi} reduces to \eqref{eq:hjb-transient-optimised}.
\end{proof}

\begin{proof}[Proof of Theorem~\ref{thm:target}]
We use the same ansatz as in the baseline model,
\[
    u_f(t,x,q,s)=x+qs+\theta_f(t,q).
\]
Substituting this into \eqref{eq:hjb-target-full} gives
\begin{equation}
\label{eq:hjb-target-theta}
    \frac{\partial \theta_f}{\partial t}(t,q)
    -
    \phi\big(q-f(t)\big)^2
    +
    \mathbf 1_{\{q>0\}}
    \sup_{\delta}
    \lambda e^{-k\delta}
    \bigg(
        \delta
        -
        \theta_f(t,q)
        +
        \theta_f(t,q-1)
        -
        \frac{\eta}{\lambda}q
    \bigg)
    =
    0,
\end{equation}
with terminal condition
\[
    \theta_f(T,q)=-\alpha q^2.
\]
The one-step optimisation is identical to the baseline case, except that the running penalty has changed. Hence the first-order condition gives
\be\label{eq:opt-aux}
    \delta_f^\star(t,q)
    =
    \frac1k
    +
    \theta_f(t,q)-\theta_f(t,q-1)
    +
    \frac{\eta}{\lambda}q .
\ee
Substituting this optimiser into \eqref{eq:hjb-target-theta} yields
\begin{equation}
\label{eq:hjb-target-theta-optimised}
    \frac{\partial \theta_f}{\partial t}(t,q)
    -
    \phi\big(q-f(t)\big)^2
    +
    \mathbf 1_{\{q>0\}}
    \frac{e^{-1}\lambda}{k}
    \exp\bigg(
        -k\bigg(
            \theta_f(t,q)-\theta_f(t,q-1)
            +
            \frac{\eta}{\lambda}q
        \bigg)
    \bigg)
    =
    0 .
\end{equation}
Now define
\[
    \omega_f(t,q):=e^{k\theta_f(t,q)}.
\]
Then \eqref{eq:hjb-target-theta-optimised} becomes the triangular linear system
\begin{equation}
\label{eq:omega-target-linear-system}
    \frac{\partial \omega_f}{\partial t}(t,q)
    =
    k\phi\big(q-f(t)\big)^2\omega_f(t,q)
    -
    \mathbf 1_{\{q>0\}}
    e^{-1}\lambda e^{-k\frac{\eta}{\lambda}q}
    \omega_f(t,q-1),
\end{equation}
with terminal condition
\[
    \omega_f(T,q)=e^{-k\alpha q^2}.
\]
Solving \eqref{eq:omega-target-linear-system} backwards from \(T\) gives precisely the recursive representation \eqref{eq:omega-target-recursive}. Finally, using \(\theta_f=(1/k)\log\omega_f\) in the optimiser \eqref{eq:opt-aux} gives \eqref{eq:optimal-quote-target}. This proves the result.
\end{proof}

\section*{Acknowledgement}
The authors are grateful to Richard Anthony (HSBC) for support throughout the project. The views expressed are those of the authors and do not necessarily reflect the views or practices at HSBC.

\bibliographystyle{abbrvnat}
\bibliography{references}

\end{document}